\documentclass[a4paper]{article}

\usepackage[height=22.5cm, width=16.5cm, hmarginratio={1:1}]{geometry}

\usepackage{amsmath}
\usepackage{amsthm}
\usepackage{bm}
\usepackage{tikz-cd}
\usetikzlibrary{decorations.markings,positioning}
\usepackage{graphicx}
\usepackage{amssymb,color}
\usepackage{mathtools}
\usepackage{appendix}
\usepackage{hyperref}

\allowdisplaybreaks[1]

\newcommand{\nn}{\nonumber}

\newcommand{\beq}{\begin{equation}}
\newcommand{\eeq}{\end{equation}}

\newcommand{\CC}{{\mathbb C}}

\newcommand{\tr}{{\rm tr}}

\newcommand{\bt}{{\bm{t}}}
\newcommand{\e}{\epsilon}

\newcommand{\He}{{\rm He}}

\newcommand{\h}{{\rm H}}

		\theoremstyle{plain}
		\numberwithin{equation}{section}
		\newtheorem{theorem}{Theorem}[section]
		\newtheorem{lemma}{Lemma}[section]
		\newtheorem{Definition}{Definition}[section]
		\newtheorem{remark}{Remark}[section]	
		\newtheorem{pro}{Proposition}[section]
		\newtheorem{cor}{Corollary}[section]

\begin{document}
\title{The matrix-resolvent method to tau-functions for the \\ nonlinear Schr\"{o}dinger hierarchy}
\author{Ang Fu, Di Yang}
\date{}
\maketitle
		
\begin{abstract}
We extend the matrix-resolvent method
of computing logarithmic derivatives of tau-functions to the nonlinear Schr\"{o}dinger (NLS) hierarchy. 
Based on this method we give a detailed  proof of a theorem of Carlet, Dubrovin and Zhang 
regarding the relationship between the Toda lattice hierarchy and the NLS hierarchy. 
As an application, we give an improvement of an algorithm of computing correlators in hermitian matrix models. 
\end{abstract}

\tableofcontents
    
\section{Introduction and statements of the results}\label{sec1}

Let
	$\mathcal{A}:=\mathbb{C}\left[q_0, r_0, q_1, r_1, q_{2}, r_{2}, \cdots\right]$ be the polynomial ring. Define a derivation $\partial:\mathcal{A}\rightarrow \mathcal{A}$ via
\begin{align}
	\partial(q_i)= q_{i+1},\quad \partial(r_i)= r_{i+1},\quad \partial(fg)=\partial(f)g+f\partial(g),\quad \forall\, i\ge 0,\,\forall\, f,g\in\mathcal{A}.
\end{align}
For convenience,  denote $q=q_0,r=r_0$. Let $\mathcal{L}^{\rm{NLS}}(\xi)$ be the matrix Lax operator (cf.~e.g.~\cite{AC91,D03, L99, NMPZ84,T04}) 
\begin{align}
	\mathcal{L}^{\rm{NLS}}(\xi)=\begin{pmatrix}
		\epsilon\partial&0\\
		0&\epsilon\partial
	\end{pmatrix}+\begin{pmatrix}
		-\xi & -q\\
		r & \xi\\
	\end{pmatrix}=\epsilon\partial+U^{\rm{NLS}}(\xi),\quad U^{\rm{NLS}}(\xi):=\begin{pmatrix}
	-\xi & -q\\
	r & \xi\\
\end{pmatrix}.\label{laxnls}
\end{align}
Here, $\xi$ is a parameter. 
We have the following lemma.
\begin{lemma}\label{lemmarnls}
	There exists a unique element $R^{\rm{NLS}}(\xi)\in\mathrm{Mat} \left(2, \mathcal{A}[\epsilon]\left[\left[\xi^{-1} \right] \right]\right)$, such that
	\begin{align}
	&	R^{\rm{NLS}}(\xi)-\begin{pmatrix}
			2&0\\
			0&0\\
		\end{pmatrix}
\in\mathrm{Mat} \left(2, \mathcal{A}[\epsilon]\left[\left[\xi^{-1} \right] \right]\xi^{-1}\right)\label{rnlsdef},\\
	&   \left[\mathcal{L}^{\rm{NLS}}(\xi), R^{\rm{NLS}}(\xi)\right]=0\label{21},\\
	&    \mathrm{tr}\ R^{\rm{NLS}}(\xi)=2, \quad \mathrm{det}\ R^{\rm{NLS}}(\xi)=0.\label{22}	
	\end{align}
\end{lemma}
The proof of this lemma is given in Section~\ref{section2}.
We call the unique series $R^{\rm{NLS}}(\xi)$ in the above Lemma~\ref{lemmarnls} the {\it basic matrix resolvent} (basic MR) of $\mathcal{L}^{\rm{NLS}}(\xi)$ (cf.~\cite{BDY16, BDY21,D81-1,D81-2,D20,Z15}). For the reader's convenience, we give the first few terms of  $R^{\rm{NLS}}(\xi)$:
\begin{align}
	     R^{\rm{NLS}}(\xi)=&
	\begin{pmatrix}
		2&0\\0&0\\
	\end{pmatrix}+
	\begin{pmatrix}
		0&q\\
		-r&0
	\end{pmatrix}\frac{1}{\xi}
	+\begin{pmatrix}
		\frac{qr}{2}&\frac{\epsilon q_1}{2}\\
		\frac{\epsilon r_1}{2}&-\frac{qr}{2}
	\end{pmatrix}\frac{1}{\xi^2}
	+\begin{pmatrix}
		\frac{\epsilon}{4}(q_1r-qr_1)&\frac{1}{4}(\epsilon^2q_{2}+2q^2r)\\
		-\frac{1}{4}(\epsilon^2r_{2}+2r^2q)&-\frac{\epsilon}{4}(q_1r-qr_1)
	\end{pmatrix}\frac{1}{\xi^3} \nonumber\\
  &+\begin{pmatrix}
  	\frac{1}{8}(\epsilon^2(qr_{2}- q_1r_1+rq_{2})+3q^2r^2)&\frac{1}{8}(\epsilon^3q_{3}+6\epsilon qq_1r)\\
  	\frac{1}{8}(\epsilon^3r_{3}+6\epsilon rr_1q)&-\frac{1}{8}(\epsilon^2(qr_{2}-q_1r_1+rq_{2})+3q^2r^2)
  \end{pmatrix}\frac{1}{\xi^4}
	+\mathcal{O}(\xi^{-5}). \label{nlsresolv}        
\end{align}

Recall that a derivation~$D$ on~$\mathcal{A}$ is called \textit{admissible}, if it commutes with~$\partial$. Let $(D_j)_{j\ge 0}$ be a sequence of 
admissible derivations, defined via
\begin{align}
	D_{j}(\mathcal{L}^{\rm{NLS}}(\xi))=2^{j}\epsilon^{-1} \Bigl[V_{j}^{\rm NLS}(\xi),\mathcal{L}^{\rm{NLS}}(\xi)\Bigr],\label{zeroeqnls}
\end{align}
where $V_{j}^{\rm NLS}(\xi)=\left(\xi^{j+1}R^{\rm NLS}(\xi)\right)_{+}$. 
Using \eqref{21}, \eqref{zeroeqnls}, we can prove that $D_j$ are well defined from~\eqref{zeroeqnls} and we have the more explicit expression
\begin{align}
	D_j(q)=2^{j+1}\epsilon^{-1}B_{j+1},\qquad
	D_j(r)=2^{j+1}\epsilon^{-1}C_{j+1}.\label{abstflow}
\end{align}
For example, $D_0(q)= q_1, D_0(r)= r_1, D_1(q)=\epsilon q_{2}+2\epsilon^{-1}q^2r, D_1(r)=\epsilon r_{2}+2\epsilon^{-1}r^2q$.
 Clearly, $D_0=\partial$. We will prove in Lemma \ref{corcom} that $(D_j)_{j\ge 0}$ all commute. We call $(D_j)_{j\ge 0}$ the \textit{NLS derivations},  and the formal system \eqref{abstflow} the \textit{abstract NLS hierarchy}. 
 
Define the {\it loop operator} $\nabla^{\rm{NLS}}(\xi)$ by
\begin{align}
     \nabla^{\rm{NLS}}(\xi):=\sum_{j\ge 0} \frac{1}{\xi^{j+2}}\frac{D_j}{2^j}.
\end{align} 
Using~\eqref{abstflow}, we have
\begin{align}\label{dr}
		\epsilon\nabla^{\rm{NLS}}(\xi)(q)=2b(\xi)-\frac{2q}{\xi},\quad
		\epsilon\nabla^{\rm{NLS}}(\xi)(r)=2c(\xi)+\frac{2r}{\xi}.
\end{align}
The following lemma is important.
\begin{lemma}\label{lemma21}
	The following equation holds true:
	\begin{gather}
		\epsilon\nabla^{\rm{NLS}}(\nu)\left(R^{\rm{NLS}}(\xi)\right)=\frac{\left[R^{\rm{NLS}}( \nu), R^{\rm{NLS}}(\xi)\right]}{\nu-\xi}+\left[Q^{\rm{NLS}}(\nu), R^{\rm{NLS}}(\xi)\right]\label{derivition},
	\end{gather}
	where $	Q^{\rm{NLS}}(\nu)=-\frac{1}{\nu}\begin{pmatrix}
		2&0\\
		0&0
	\end{pmatrix}$.
\end{lemma}	
The proof is in Section~\ref{section2}. 

The purpose of our study is to use the matrix-resolvent method 
\cite{BDY16,BDY21,DY17,Z15} to the study of 
tau-structure for the abstract NLS hierarchy and so for the NLS hierarchy (see below). 
First, let us use matrix resolvent to the 
definition of tau-structure. Observe that
\begin{align*}
	\frac{\mathrm{tr}\, R^{\rm{NLS}}(\xi)R^{\rm{NLS}}(\nu)-4}{(\xi-\nu)^2}\in\mathcal{A}[\epsilon]\left[\left[\xi^{-1},\nu^{-1}\right]\right]\xi^{-2}\nu^{-2}.
\end{align*}
Indeed,
\begin{gather}
	\mathrm{tr}\ R^{\rm{NLS}}(\xi)R^{\rm{NLS}}(\nu)-4 \label{27}
\end{gather}
is divisible by $(\xi-\nu)^2$. By using \eqref{22},  we have $\mathrm{tr}\left(R^{\rm{NLS}}(\xi)^2\right)=4$. It implies that \eqref{27} is divisible by $(\xi-\nu)$. Because the dependence of this expression is symmetric with respect to changing $\xi,\nu$,  this implies the divisibility by $(\xi-\nu)^2$.
\begin{Definition}\label{lemmataustr}
	Define a family of elements $\Omega_{i,j}^{\rm NLS}\in\mathcal{A}[\epsilon],\,i,j\ge 0$, via the generating series
	\begin{gather}
		\sum_{i, j\ge 0}\frac{1}{\xi^{i+2}\nu^{j+2}}\frac{\Omega_{i,j}^{\rm NLS}}{2^{i+j}}=\frac{\mathrm{tr}\, R^{\rm{NLS}}(\xi)R^{\rm{NLS}}(\nu)-4}{(\xi-\nu)^2}.   \label{tastr}
	\end{gather}
\end{Definition}
\begin{lemma}\label{2taustr}
	The differential polynomials $\Omega_{i,j}^{\rm NLS}$, $i,j\ge 0$, have the following properties:
	\begin{align}
		\Omega_{0,0}^{\rm NLS}=qr,\quad \Omega_{i,j}^{\rm NLS}=\Omega_{j,i}^{\rm NLS},\quad D_k(\Omega_{i,j}^{\rm NLS})=D_j(\Omega_{i,k}^{\rm NLS}),\quad \forall\, i,j,k\ge 0.\label{eqtaustr}
	\end{align}
\end{lemma}
The proof of this lemma is given in Section~\ref{section2}.

We call $(\Omega_{i,j}^{\rm NLS})_{i,j\ge 0}$ \textit{the tau-structure}~\cite{CDZ04,DZ,DZ04} for the abstract NLS hierarchy~\eqref{abstflow}. 
The first few  $\Omega_{i,j}^{\rm NLS}$ are
\begin{align}
	&\Omega_{0,0}^{\rm NLS}=qr,\quad\Omega_{1,0}^{\rm NLS}=\epsilon(rq_1-qr_1),\quad \Omega_{1,1}^{\rm NLS}=2q^2r^2+\epsilon^2(qr_2+q_2r-2q_1r_1),\\
	&\Omega_{2,0}^{\rm NLS}=3q^2r^2+\epsilon^2(-q_1r_1+rq_2+qr_2),\\
	&\Omega_{2,1}^{\rm NLS}=\epsilon(-6 q^2rr_1+6r^2qq_1)-2r_1q_2\epsilon^2+\epsilon^3(2q_1r_2+rq_3-qr_3).
\end{align}
For $k\geq3$, define  
\begin{align}
	\Omega_{i_1,\dots,i_k}^{\rm NLS}:=D_{i_1}\cdots D_{i_{k-2}}\left(\Omega_{i_{k-1},i_k}^{\rm NLS}\right) \in \e^{k-2}\mathcal{A}[\epsilon],\quad i_1,\dots,i_k\ge0.  \label{symmetric}
\end{align}
By using \eqref{eqtaustr} we know that the $\Omega_{i_1,\dots,i_k}^{\rm NLS}$, $k\ge 2$, are totally symmetric with respect to permutations of the indices $i_1,\dots,i_k$.
We then have the following theorem.
\begin{theorem}[\cite{D20}] \label{npointc}
For any integer $k\ge 3$, we have
    \begin{align}
		\sum_{i_1,\dots,i_k\ge 0} \Omega_{i_1,\dots,i_k}^{\rm NLS}\prod_{j=1}^{k}\frac{1}{2^{i_j}\xi_{j}^{i_j+2}}=-\sum_{\sigma\in S_k/ C_k} \frac{\mathrm{tr}\,[R^{\rm{NLS}}(\xi_{\sigma(1)},\epsilon)\cdots R^{\rm{NLS}}(\xi_{\sigma(k)},\epsilon)]}{(\xi_{\sigma(1)}-\xi_{\sigma(2)})\cdots(\xi_{\sigma(k-1)}-\xi_{\sigma(k)})(\xi_{\sigma(k)}-\xi_{\sigma(1)})}\label{ntaustr},
	\end{align}
where $S_k$ denotes the symmetric group, $C_k$ the cyclic group, and it is understood that $\sigma(k+1)=\sigma(1)$.
\end{theorem}
The proof of this theorem is in Section~\ref{section2}.

If we think of $q=q(X,{\bf t};\epsilon),r=r(X,{\bf t};\epsilon)$ as functions of $X$, ${\bf t}=(t_0,t_1,t_{2},\dots)$,
and identify $\partial$ with $\partial_{X}$, $q_i,r_i$ via  
\begin{align}
	q_i\mapsto\partial_{X}^{i}(q(X,{\bf t};\epsilon)),\quad r_i\mapsto\partial_{X}^{i}(r(X,{\bf t};\epsilon)),\quad i\ge 0,\label{Tsubs}
\end{align} 
and $D_j$ with  
$\partial/\partial t_{j}$, then the abstract NLS hierarchy~\eqref{abstflow} becomes
\begin{align}\label{Tnlshierarchy}
	\e\frac{\partial q}{\partial t_{j}}=2^{j+1} B_{j+1},\quad
	\e\frac{\partial r}{\partial t_{j}}=2^{j+1} C_{j+1},\quad j\ge 0. 
\end{align}
Equations~\eqref{Tnlshierarchy} are called 
the {\it AKNS hierarchy} ({\it aka} the {\it NLS hierarchy}) (cf.~\cite{AC91,D03,GH03,L99,T04}). 
The first few of them are
\begin{align}
	& q_{t_{0}}=q_X,\quad r_{t_{0}}=r_X,\label{t0=x}\\
	& q_{t_1}=\epsilon q_{XX}+2\epsilon^{-1}q^2r,\quad  r_{t_1}=-\epsilon r_{XX}-2\epsilon^{-1}r^2q,\\
	& q_{t_2}=\epsilon^2q_{XXX}+6 qq_Xr,\quad  r_{t_2}=\epsilon^2r_{XXX}+6 rr_Xq.
\end{align}
Because of~\eqref{t0=x}, we identify~$t_{0}$ with~$X$ and write $q(X,{\bf t};\epsilon),r(X,{\bf t};\epsilon)$ simply as $q({\bf t};\epsilon), r({\bf t};\epsilon)$. 

Let $(q({\bf t};\epsilon),r({\bf t};\epsilon))$ be a solution to the NLS hierarchy~\eqref{Tnlshierarchy}. 
For $k\geq2$, write $\Omega_{i_1,\dots,i_k}^{\rm NLS}({\bf t},\epsilon)$ as the image of $\Omega_{i_1,\dots,i_k}^{\rm NLS}$ under~\eqref{Tsubs}. It then follows from Lemma~\ref{2taustr} that
 there exists a function $\tau^{\rm NLS}({\bf t};\epsilon)$, such that
\begin{align}
	\Omega_{i,j}^{\rm NLS}({\bf t};\epsilon)=\epsilon^2\frac{\partial^2 \log\tau^{\rm NLS}({\bf t};\epsilon)}{\partial t_{i}\partial t_{j}},\quad i,j\ge  0.\label{DZtauf}
\end{align}
We call $\tau^{\rm NLS}({\bf t};\epsilon)$ the \textit{Dubrovin-Zhang type tau-function} 
of the solution $(q({\bf t};\epsilon),r({\bf t};\epsilon))$ to the NLS hierarchy. 
The function $\tau^{\rm NLS}({\bf t};\epsilon)$ is determined uniquely by the solution $(q({\bf t};\epsilon),r({\bf t};\epsilon))$ up to 
multiplying by the exponential of a linear function
\begin{align}
	\tau^{\rm NLS}({\bf t};\epsilon)\mapsto e^{a_0+\sum_{j\ge 0} a_{j+1}t_{j}}\tau^{\rm NLS}({\bf t};\epsilon), \quad a_{0},a_{1},a_{2},\dots \in 
	\mathbb{C}((\epsilon)).
\end{align}

From the definition we know that  
\begin{align}
	\Omega_{i_1,\dots,i_k}^{\rm NLS}({\bf t};\epsilon) 
	=\epsilon^k\frac{\partial^k\log\tau^{\rm NLS}({\bf t};\epsilon)}{\partial t_{i_1}\cdots\partial t_{i_k}},\quad k\ge 2,\,i_1,\dots,i_k\ge 0.
\end{align}
Denote $\Omega_{i}^{\rm NLS}({\bf t};\epsilon):=\epsilon\partial_{t_{i}}(\log\tau({\bf t};\epsilon))$. 
 Let $R^{\rm{NLS}}({\bf t};\xi;\epsilon)$ denote the image of $R^{\rm{NLS}}(\xi)$ under \eqref{Tsubs}. The following corollary follows from Definition \ref{lemmataustr} and Theorem \ref{npointc}.

\begin{cor}\label{cor2.2}
For any $k\ge 2$, the following formula holds  true:
    \begin{align}
		&\sum_{i_1,\dots,i_k\ge 0} \epsilon^k\frac{\partial^k\log\tau^{\rm NLS}({\bf t};\epsilon)}{\partial t_{i_1}\cdots\partial t_{i_k}}\prod_{j=1}^{k}\frac{1}{2^{i_j}\xi_{j}^{i_j+2}} \nonumber\\
		=&-\sum_{\sigma\in S_k/ C_k}\frac{\mathrm{tr} \left(R^{\rm{NLS}}({\bf t};\xi_{\sigma(1)};\epsilon)\cdots R^{\rm{NLS}}({\bf t};\xi_{\sigma(k)};\epsilon)\right)}{(\xi_{\sigma(1)}-\xi_{\sigma(2)})\cdots(\xi_{\sigma(k-1)}-\xi_{\sigma(k)})(\xi_{\sigma(k)}-\xi_{\sigma(1)})}-\frac{4\delta_{k,2}}{(\xi_1-\xi_2)^2}.\label{cor2,21}
	\end{align}
\end{cor}

Using the matrix-resolvent method 
to tau-functions for the NLS hierarchy and for 
the Toda lattice hierarchy~\cite{DY17}, we will give a detailed proof 
of a theorem given by G.~Carlet, B.~Dubrovin and Y.~Zhang \cite{CDZ04} (cf.~also~\cite{CLPS21}), written into the 
following two parts.
(A brief review of the matrix-resolvent method to tau-functions for the Toda lattice hierarchy 
can be found in Appendix~\ref{appa}.)

\begin{theorem} [Carlet--Dubrovin--Zhang~\cite{CDZ04}] \label{part1}
Let $(v(x;{\bf t};\epsilon), w(x;{\bf t};\epsilon))$ be an arbitrary solution to the Toda lattice hierarchy, and $\tau(x;{\bf t};\epsilon)$ the tau-function of the solution $(v(x;{\bf t};\epsilon), w(x;{\bf t};\epsilon))$. Define
      	\begin{align}
      	q(x;{\bf t};\epsilon):=\frac{\tau(x+\epsilon;{\bf t};\epsilon)}{\tau(x;{\bf t};\epsilon)},\qquad
      	r(x;{\bf t};\epsilon):=\frac{\tau(x-\epsilon;{\bf t};\epsilon)}{\tau(x;{\bf t};\epsilon)}.\label{importandef}
      	\end{align}
      	Then, for any $x,\epsilon$, $(q=q(x;{\bf t};\epsilon),r=r(x;{\bf t};\epsilon))$ is a solution to the NLS hierarchy \eqref{Tnlshierarchy} with $\partial_{t_{0}}$ being identified with~$\partial_X$, and
	 $\tau(x;{\bf t};\epsilon)$ is the tau-function of the solution $(q(x;{\bf t};\epsilon),r(x;{\bf t};\epsilon))$ to the NLS hierarchy.
\end{theorem}

\begin{theorem} [Carlet--Dubrovin--Zhang~\cite{CDZ04}] \label{part2}
Let $(q({\bf t};\epsilon),r({\bf t};\epsilon))$ be an arbitrary power series in $t_{>0}$ solution to the NLS hierarchy~\eqref{Tnlshierarchy} such that 
$\epsilon\log q(\bt;\e)$ and $q(\bt;\e) r(\bt;\e)$ are power series of~$\epsilon$.
Define 
\begin{align}\label{defqr*}
    V({\bf t};\epsilon):=\epsilon\frac{\partial_X(q({\bf t};\epsilon))}{q({\bf t};\epsilon)},\quad
	W({\bf t};\epsilon):=q({\bf t};\epsilon)r({\bf t};\epsilon).
\end{align}
Assume that the following difference equation for $(v,w)$
\begin{align}\label{todaeq1}
\begin{split}
&(\Lambda-1)(w)=\epsilon\partial_{X}(v), \quad w \, (1-\Lambda^{-1})(v)=\epsilon \, \partial_{X}(w)
\end{split}
\end{align}
with the initial condition 
\begin{align}
	v(x;{\bf t};\epsilon)|_{x=0}=V({\bf t};\epsilon),\quad w(x;{\bf t};\epsilon)|_{x=0}=W({\bf t};\epsilon).\label{todaeqini2}
\end{align}
has a unique 
power-series-in-$(x, t_{>0})$ solution 
$(v(x;{\bf t};\epsilon),w(x;{\bf t};\epsilon))$.
Then $(v(x;{\bf t};\epsilon),w(x;{\bf t};\epsilon))$ satisfies the Toda lattice hierarchy.
Furthermore, define $q=q(x;{\bf t};\epsilon)$ as the unique function (up to a constant that may depend on~$\e$)  
satisfying
\begin{align}
	&(1-\Lambda^{-1})\log q=\log w, \label{defw}\\ 
	&\e \, \frac{\partial \log q}{\partial t_{k}}=\Lambda (c_{k+1}),\quad k\ge 0,\label{defqx}
\end{align} 
with $c_{k+1}$ defined in~\eqref{compoR}, 
and define 
$r(x;{\bf t};\epsilon):=1/q(x-\epsilon;{\bf t};\epsilon)$, then 
there exists a function $\tau^{\rm NLS}(x;{\bf t};\epsilon)$ satisfying 
\begin{align}
&\epsilon^2\frac{\partial^2 \log\tau^{\rm NLS}(x;{\bf t};\epsilon)}{\partial t_{i}\partial t_{j}} = 
	\Omega_{i,j}^{\rm NLS}(x;{\bf t};\epsilon),\quad i,j\ge  0,\label{deftau1}\\
&\epsilon (\Lambda-1)\frac{\partial \log \tau^{\rm NLS}(x;{\bf t};\epsilon)}{\partial t_{k}}=c_{k+1}(x+\epsilon;{\bf t};\epsilon),\quad k\ge 0,\label{deftau2}\\
&(\Lambda+\Lambda^{-1}-2)\log \tau^{\rm NLS}(x;{\bf t};\epsilon)=\log (q(x;{\bf t};\epsilon)r(x;{\bf t};\epsilon));\label{deftau3}
\end{align}
moreover, $\tau^{\rm NLS}(x;{\bf t};\epsilon)$ is the tau-function of the solution $(v(x;{\bf t};\epsilon),w(x;{\bf t};\epsilon))$ to the Toda lattice hierarchy.
\end{theorem}

Theorems~\ref{part1}, \ref{part2}, and the MR method give rise to a simple algorithm of computing correlators in hermitian matrix models; see Section~\ref{section4} for the details and explicit examples.
	
The paper is organized as follows. In Section~\ref{section2}, we use the MR method to the study of tau-functions for the NLS hierarchy and prove Theorem~\ref{npointc}. 
In Section~\ref{sec2}, we prove Theorems \ref{part1},\ref{part2}. In Section~\ref{section4}, we develop an algorithm of computing 
correlators in hermitian matrix models.

\section{Proofs of Lemmas \ref{lemmarnls}--\ref{2taustr} }\label{section2}
In this section, we apply the MR method in studying the tau-structure for the NLS hierarchy. 

\begin{proof}[Proof of Lemma~\ref{lemmarnls}] 
Write
\begin{align}
&   R^{\rm{NLS}}(\xi)=\begin{pmatrix}
2+a(\xi)&b(\xi)\\
c(\xi)&-a(\xi)
\end{pmatrix},\\
&   a(\xi)=\sum_{j\ge 0} \frac{A_j}{\xi^{j+1}},\quad 
b(\xi)=\sum_{j\ge 0} \frac{B_j}{\xi^{j+1}}, \quad
c(\xi)=\sum_{j\ge 0} \frac{C_j}{\xi^{j+1}},		
\end{align}
where $A_j,B_j,C_j\in\mathcal{A}[\epsilon]$. In terms of $ a(\xi), b(\xi), c(\xi)$, equation~\eqref{21} and the second equation in~\eqref{22} read
\begin{align}
&		\epsilon\partial(a(\xi))-rb(\xi)-qc(\xi)=0, \label{l21}\\
&		\epsilon\partial(b(\xi))+2qa(\xi)-2\xi b(\xi)+2q=0 ,\label{l22}\\
&		\epsilon\partial (c(\xi))+2ra(\xi)+2\xi c(\xi)+2r=0 \label{l23},\\
&       a^2(\xi)+2a(\xi)+b(\xi)c(\xi)=0.\label{nomal}
\end{align}
These equations lead to relations for $A_j,B_j,C_j$:
\begin{align}
&	\epsilon\partial(A_{j})-rB_j-qC_j=0,\quad j\ge 0, \label{abc1}\\
&	\epsilon\partial(B_{j})+2qA_j-2B_{j+1}+2q\delta_{-1,j}=0,\quad j\ge -1,\label{abc2}\\
&	\epsilon\partial(C_{j})+2rA_j+2C_{j+1}+2r\delta_{-1,j}=0,\quad j\ge -1, \label{abc3}\\
&      A_k=-\frac{1}{2}\sum_{\substack{ i+j=k-1\\ i,j\ge -1}} (A_iA_j+B_iC_j), \qquad k\geq 0.\label{23}
\end{align}
Here $\delta_{i,j}$ denotes the Kronecker delta and $A_{-1}=B_{-1}=C_{-1}:=0$. It is clear from~\eqref{abc1}--\eqref{23} that if $ a(\xi), b(\xi), c(\xi)$ exist then they must be unique. For the existence of $R^{\rm{NLS}}(\xi)$, we need to prove that \eqref{l21}--\eqref{nomal} are compatible. Indeed, applying the derivation $\partial$ on both sides of \eqref{nomal}, we obtain
\begin{align}
		2a(\xi)\partial(a(\xi))+2\partial(a(\xi))+\partial(b(\xi))c(\xi)+b(\xi)\partial(c(\xi))=0, \label{24}
\end{align}
which agrees with~\eqref{l21}--\eqref{l23}. The lemma is proved.          
\end{proof}

Following the uniqueness argument of~\cite{Y20}, 
in order to prove Lemma~\ref{lemma21}, we will first prove the following lemma.
\begin{lemma}\label{wlemma}
There exists a unique element $W(\xi,\nu)$ in $\mathcal{A}[\epsilon]\otimes sl_2(\mathbb{C})[[\xi^{-1},\nu^{-1}]]\xi^{-1}\nu^{-1}$ of the form
\begin{align}
		W(\xi,\nu)=\begin{pmatrix}
			X(\xi,\nu)&Y(\xi,\nu)\\
			Z(\xi,\nu)&-X(\xi,\nu)
		\end{pmatrix}
\end{align}
satisfying the following two equations for $W(\xi,\nu):$
\begin{align}
		&\epsilon\partial(W(\xi,\nu))+\left[U^{\rm{NLS}}(\xi),W(\xi,\nu)\right]+\left[\nabla^{\rm{NLS}}(\nu)(U^{\rm{NLS}}(\xi)),R^{\rm{NLS}}(\xi)\right]=0, \label{wlemma1}\\
		&2X(\xi,\nu)+2a(\xi)X(\xi,\nu)+c(\xi)Y(\xi,\nu)+b(\xi)Z(\xi,\nu)=0.\label{wlemma2}
\end{align}
\end{lemma}
\begin{proof}[Proof]
The existence part of this lemma follows from Lemma \ref{lemmarnls}. 
Indeed, denote
\begin{align}
		W(\xi,\nu):=\nabla^{\rm{NLS}}(\nu) \bigl(R^{\rm{NLS}}(\xi)\bigr).
\end{align}
Then it follows from Lemma~\ref{lemmarnls} that 
$W(\xi,\nu)$ belongs to $\mathcal{A}[\epsilon]\otimes sl_2(\mathbb{C})[[\xi^{-1},\nu^{-1}]]\xi^{-1}\nu^{-1}$ and  satisfies \eqref{wlemma1}--\eqref{wlemma2}. To see the uniqueness part, write 
\begin{align}
		X(\xi,\nu)=\sum_{i,j\ge 0} \frac{X_{i,j}}{\xi^{i+1}\nu^{j+1}},\quad Y(\xi,\nu)=\sum_{i,j\ge 0} \frac{Y_{i,j}}{\xi^{i+1}\nu^{j+1}},\quad Z(\xi,\nu)=\sum_{i,j\ge 0} \frac{Z_{i,j}}{\xi^{i+1}\nu^{j+1}}, 
\end{align}
where $X_{i,j}, Y_{i,j}, Z_{i,j}\in \mathcal{A}[\epsilon],  i,j\ge 0$.  Substituting these expressions into \eqref{wlemma1}--\eqref{wlemma2}, we find that these equalities give rise to the following recurison relations for $X_{i,j}, Y_{i,j}, Z_{i,j}$ :
\begin{align}
&		\epsilon\partial(X_{i,j})=-\left(2q\delta_{0,j}-2B_j\right)C_i+\left(2r\delta_{0,j}+2C_j\right)B_i+rY_{i,j}+qZ_{i,j}\label{pw1},\\
&		\epsilon\partial(Y_{i,j})=2A_i\left(2q\delta_{0,j}-2B_j\right)-2qX_{i,j}+2Y_{i+1,j}\label{pw2},\\
&		\epsilon\partial(Z_{i,j})=-2A_j\left(2r\delta_{0,j}+2C_j\right)-2rX_{i,j}-2Z_{i+1,j},\\
&		X_{i,j}=-\frac{1}{2}\sum_{l+m=i-1} (2A_lX_{m,j}+C_lY_{m,j}+B_lZ_{m,j}),\label{pw4}\\
&		X_{0,j}=0,\quad Y_{0,j}=-2q\delta_{0,j}+2B_j,\quad Z_{0,j}=-2r\delta_{0,j}-2C_j,\label{pw5}
\end{align}
where $i\ge 1$, $j\ge 0$. 
Using these recursion relations we see that $W(\xi,\nu)$ are uniquely determined. The lemma is proved.	
\end{proof}
We are now ready to prove Lemma \ref{lemma21}.
\begin{proof}[Proof of Lemma \ref{lemma21}]
Define $W^*$ as the right-hand side of $\eqref{derivition}$. Write $W^*=\epsilon^{-1}\begin{pmatrix}
		X^*&Y^*\\
		Z^*&-X^*
	\end{pmatrix}$.
Explicitly,
\begin{align}
		X^*&=\frac{-b(\xi)c(\nu)+b(\nu)c(\xi)}{\nu-\xi},\label{1}\\
		Y^*&=-\frac{2b(\xi)}{\nu}+\frac{-2(1+a(\xi))b(\nu)+2(1+a(\nu))b(\xi)}{\nu-\xi},\\
		Z^*&=\frac{2c(\xi)}{\nu}+\frac{2(1+a(\xi))c(\nu)-2(1+a(\nu))c(\xi)}{\nu-\xi}\label{3}.
\end{align}
Then it is straightforward to verify that $W^* \in \mathcal{A}[\epsilon]\otimes sl_2(\mathbb{C})[[\xi^{-1},\nu^{-1}]]\xi^{-1}\nu^{-1}$.
We can also verify that $W := W ^*$ satisfies \eqref{wlemma1}--\eqref{wlemma2}. Indeed, by using \eqref{dr}, \eqref{abc1}--\eqref{abc3} and \eqref{1}--\eqref{3},  we have
\begin{align*}
		\epsilon\partial X^*=&\frac{2}{\xi-\nu}(b(\nu) (r +  r a(\xi) +  (\xi-\nu) c(\xi)) - b(\xi) (r +  r a(\nu) -  (\xi-\nu) c(\nu))  \\
		&+q ((1 +  a(\nu)) c(\xi) - (1 +  a(\xi)) c(\nu))),\\
		\epsilon\partial Y^*=&-\frac{2}{\nu(\nu-\xi)}(2q(\xi-\nu) + 2 \nu^2 b(\nu) + 2 a(\xi) (q (\xi - \nu) +  \nu^2 b(\nu)) +  q \nu b(\nu) c(\xi)\\
		& -  b(\xi) (2\xi^2 + 2 \xi \nu a(\nu) + q \nu c(\nu))),\\
		\epsilon\partial Z^*=&\frac{2}{\nu(\nu-\xi)}(2r(\xi-\nu) - 2 \nu^2 c(\nu) +  2 a(\xi) (r (\xi - \nu) - \nu^2 c(\nu)) + r \nu c(\nu) b(\xi)\\
		& + c(\xi) (2\xi^2 + 2 \xi \nu a(\nu) - r \nu b(\nu))).
\end{align*}
Substituting these into the left-hand side of \eqref{wlemma1}--\eqref{wlemma2} with $W := W ^*$,  we find that \eqref{wlemma1}--\eqref{wlemma2} hold  true.  The lemma is proved due to Lemma \ref{wlemma}.
\end{proof}

In the next lemma, let us prove that the admissible derivations $(D_j)_{j\ge 0}$, defined in~\eqref{zeroeqnls}, all commute. 
\begin{lemma}\label{commute}\label{corcom}
The derivations $(D_j)_{j\ge 0}$ commute pairwise.
\end{lemma}
\begin{proof}[Proof]
Since $(D_j)_{j\ge 0}$ are admissible derivations, it suffices to prove that
\begin{align}
		D_iD_j(q)=D_jD_i(q),\quad D_iD_j(r)=D_jD_i(r),\quad \forall\,\, i,j\ge 0.
\end{align}
Without loss of generality, assume that $i<j$. Comparing the coefficient $\xi^{-i-2}\nu^{-j-2}$ of \eqref{derivition}, we have for example
\begin{align}
		D_jD_i(q)&=2^{i+1}\epsilon^{-1}D_j(B_{i+1})=2^{i+j+2}\epsilon^{-2}\left(B_{i+j+2}+\sum_{k=0}^{j}(A_kB_{i+j-k+1}-A_{i+j-k+1}B_k)\right),\\
		D_iD_j(q)&=2^{j+1}\epsilon^{-1}D_i(B_{j+1})=2^{i+j+2}\epsilon^{-2}\left(B_{i+j+2}+\sum_{k=0}^{i}(A_kB_{i+j-k+1}-A_{i+j-k+1}B_k)\right).
\end{align}
Therefore,
\begin{align*}
		D_jD_i(q)-D_iD_j(q)&=2^{i+j+2}\epsilon^{-2}\sum_{k=i+1}^{j}(A_kB_{i+j-k+1}-A_{i+j-k+1}B_k)
		=0.
\end{align*} 
Similarly, we have $D_jD_i(r)=D_iD_j(r)$. The lemma is proved. 
\end{proof}

\begin{proof}[Proof of Lemma~\ref{2taustr}]
The proof is similar to that in~\cite{BDY21,DY17}. We omit its details.
\end{proof}

\begin{proof}[Proof of Theorem~\ref{npointc}]
Based on Lemmas~\ref{lemmarnls},~\ref{lemma21} and Definition~\ref{lemmataustr}, the proof is then similar to that in~~\cite{BDY16,BDY21}, and we therefore omit its details.
\end{proof}

Let $q({\bf t};\epsilon),r({\bf t};\epsilon)$ be a solution to the NLS hierarchy \eqref{Tnlshierarchy}, and $\tau^{\rm NLS}({\bf t};\epsilon)$ the tau-function of the solution. 

The following corollary follows from Corollary~\ref{cor2.2}. Indeed, it is clear that we can take $t_{>0}=0$ on both sides of~\eqref{cor2,21}. Denote for short $R^{\rm NLS}(X;\xi;\epsilon)=R^{\rm NLS}({\bf t};\xi;\epsilon)|_{t_{>0}=0}$.

\begin{cor}\label{cor2.3}
For any $k\ge 2$, the following formula holds true:
\begin{align}
&\sum_{i_1,\dots,i_k\ge 0} \epsilon^k\frac{\partial^k\log\tau^{\rm NLS}({\bf t};\epsilon)}{\partial t_{i_1}\cdots\partial t_{i_k}}\bigg|_{t_{>0}=0}\prod_{j=1}^{k}\frac{1}{2^{i_j}\xi_{j}^{i_j+2}} \nonumber\\
=&-\sum_{\sigma\in S_k/ C_k} \frac{\mathrm{tr}\,[R^{\rm{NLS}}(X;\xi_{\sigma(1)};\epsilon)\cdots R^{\rm{NLS}}(X;\xi_{\sigma(k)};\epsilon)]}{(\xi_{\sigma(1)}-\xi_{\sigma(2)})\cdots(\xi_{\sigma(k-1)}-\xi_{\sigma(k)})(\xi_{\sigma(k)}-\xi_{\sigma(1)})}-\frac{4\delta_{k,2}}{(\xi_1-\xi_2)^2}.\label{cor2.22}
\end{align}
\end{cor}

Corollary \ref{cor2.3} gives an algorithm with the initial value $(f(X;\epsilon),g(X;\epsilon))$ of the solution $(q({\bf t};\epsilon),r({\bf t};\epsilon))$ as the only input for computing the $k_{\rm th}$-order logarithmic derivatives of the tau-function $\tau^{\rm NLS}({\bf t};\epsilon)$ evaluated at $t_{>0}=0$ for $k\ge 2$. 
By using Lemma \ref{lemmarnls}, $R^{\rm NLS}(X;\xi;\epsilon)$ is determined uniquely by the recurrence relations \eqref{abc1}--\eqref{23} and this gives an effective way to compute the matrix resolvent $R^{\rm NLS}(X;\xi;\epsilon)$;  the coefficients in the ${\bf t}$-expansion of $\tau^{\rm NLS}({\bf t};\epsilon)$ are then obtained through algebraic manipulations by using~\eqref{cor2.22}.

\section{Proofs of Theorems \ref{part1}, \ref{part2}}\label{sec2} 

In this section, we prove Theorems \ref{part1},\ref{part2}. In the proofs, we will  use the 
MR method for the NLS hierarchy and for the Toda lattice 
hierarchy (cf.~\eqref{laxtoda}--\eqref{todatauf3} in Appendix~\ref{appa}).

\begin{proof} [Proof of Theorem~\ref{part1}]
By using \eqref{211}--\eqref{todatauf3} and \eqref{importandef}, we have
\begin{align}
 & \partial_{t_{j}}(\log q(x;{\bf t};\epsilon))=\epsilon^{-1}c_{j+1}(x+\epsilon;{\bf t};\epsilon),\quad \partial_{t_{j}}(\log\,r(x;{\bf t};\epsilon))=-\epsilon^{-1}c_{j+1}(x;{\bf t};\epsilon),\quad j\ge 0.   \label{tildeDj}	\\
 & q(x;{\bf t};\epsilon)r(x;{\bf t};\epsilon)=w(x;{\bf t};\epsilon),\quad q(x-\epsilon;{\bf t};\epsilon)r(x;{\bf t};\epsilon)=1. \label{qrrelation}
\end{align}
 In particular, when $j=0$, we have
\begin{align}
	\partial_{X}(q(x;{\bf t};\epsilon))=\epsilon^{-1}q(x;{\bf t};\epsilon)v(x;{\bf t};\epsilon),\quad\partial_{X}(r(x;{\bf t};\epsilon))=-\epsilon^{-1}r(x;{\bf t};\epsilon)v(x-\epsilon;{\bf t};\epsilon).\label{0flow}
\end{align}
Let us prove two lemmas.
\begin{lemma}
We have the following identity
\begin{align}\label{keylaxrelation}
T(x;{\bf t};\epsilon)^{-1}\circ\mathcal{L}^{\rm NLS}(\xi)\mid_{q=q(x;{\bf t};\epsilon),r=r(x;{\bf t};\epsilon)}\circ T(x;{\bf t};\epsilon)=\mathcal{A}(\lambda)+\frac{\lambda}{2}{\rm I},\quad \lambda=2\xi,
\end{align}
where 
$T(x;{\bf t};\epsilon):=\begin{pmatrix}
-1&0\\
0&r(x;{\bf t};\epsilon)
\end{pmatrix}$ 
and 
$\mathcal{A}(\lambda):=
\epsilon\partial_{X}-\begin{pmatrix} \lambda & -w(x;\bt; \e) \\ 1 & v(x-\e; \bt; \e)\end{pmatrix}$. 
\end{lemma}
\begin{proof}
By using \eqref{qrrelation} and \eqref{0flow} with a direct computation, the lemma is proved.
\end{proof}
Denote for short $\mathcal{L}^{\rm NLS}(\xi)=\mathcal{L}^{\rm NLS}(\xi)|_{q=q(x;{\bf t};\epsilon),r=r(x;{\bf t};\epsilon)}$.
\begin{lemma} \label{thmReso} 
The following identity holds:
 \begin{align}\label{matrixrelation}
 R^{\rm NLS}(x;{\bf t};\xi;\epsilon)=2T(x;{\bf t};\epsilon)R(x;{\bf t};\lambda;\epsilon)T^{-1}(x;{\bf t};\epsilon).
 \end{align}
\end{lemma} 
\begin{proof}
Denote $\widehat{R}(\lambda)=2T(x;{\bf t};\epsilon)R(x;{\bf t};\lambda;\epsilon)T^{-1}(x;{\bf t};\epsilon)$. By using~\eqref{keylaxrelation}, we have
\begin{align}
  \left[\mathcal{L}^{\rm NLS}(\xi),\widehat{R}(\lambda)\right]=T[\mathcal{A}(\lambda),2R(x;{\bf t};\lambda;\epsilon)]T^{-1}.
\end{align}
By using ${\rm det}\, R(x;{\bf t};\xi;\epsilon)=0, {\rm tr} \, R(x;{\bf t};\xi;\epsilon)=1$, \eqref{Rleadingterm1023}
and~\eqref{35}, we know that 
\begin{align}
\left[\mathcal{L}^{\rm NLS}(\xi),\widehat{R}(\lambda)\right]=0, 
\quad {\rm det} \, \widehat{R}(\lambda)=0,\quad {\rm tr}\, \widehat{R}(\lambda)=2,\quad
 \widehat{R}(\lambda)=\begin{pmatrix}2&0\\0&0\end{pmatrix}
 +\mathcal{O}\left(\xi^{-1}\right). \label{rstar1023}
\end{align}
By definition, $R^{\rm NLS}(x;{\bf t};\xi;\epsilon)$ also satisfies~\eqref{rstar1023}. 
It is clear from the proof of Lemma~\ref{lemmarnls} that the solution to~\eqref{rstar1023} 
is unique. Therefore,  
$\widehat{R}(\lambda)=R^{\rm NLS}(x;{\bf t};\xi;\epsilon)$. The lemma is proved.
\end{proof}

Let us now continue the proof of the theorem. By using  \eqref{todarec1}, \eqref{absfloetoda}, \eqref{35} and \eqref{dcj},  we obtain
\begin{align}
		\frac{\partial \mathcal{A}(\lambda)}{\partial t_{k}}=\epsilon^{-1}[V_{k}^{ }(\lambda),\mathcal{A}(\lambda)],\quad k\ge 0.\label{pequation}
\end{align}
Here $V _{k}(\lambda)=\Big(\lambda^{k+1}R(x;{\bf t};\lambda;\epsilon)\Big)_{+}+\begin{pmatrix}
		0&0\\
		0&c_{k+1}(x;{\bf t};\epsilon)
	\end{pmatrix}.$
Then by using \eqref{tildeDj}, \eqref{keylaxrelation}, \eqref{matrixrelation} and \eqref{pequation}, for any $k\ge 0$, we obtain
\begin{align*}
 \frac{\partial \mathcal{L}^{\rm NLS}(\xi)}{\partial t_{k}}
 =T\left(\Big[T^{-1}\frac{\partial T}{\partial t_{k}},\mathcal{A}(\lambda)\Big]+\frac{\partial \mathcal{A}(\lambda)}{\partial t_{k}}\right)T^{-1}
 =\epsilon^{-1}2^k \left[\big(\xi^{k+1}R^{\rm NLS}(x;{\bf t};\xi;\epsilon)\big)_{+},\mathcal{L}^{\rm NLS}(\xi)\right].
\end{align*}
Then by using \eqref{zeroeqnls}, for any  $x,\epsilon$, $(q(x;{\bf t};\epsilon),r(x;{\bf t};\epsilon))$ is a solution to the NLS hierarchy \eqref{Tnlshierarchy} under identifying $\partial_{t_{0}}$ with $\partial_X$. 
Substituting \eqref{matrixrelation} into \eqref{tastr}, and using \eqref{taustr} and \eqref{todatauf1}, we have
\begin{align}
\sum_{i, j\ge 0} \frac{1}{\xi^{i+2}\nu^{j+2}}\frac{\Omega_{i,j}^{\rm NLS}(x;{\bf t};\epsilon)}{2^{i+j}}=\frac{4(\mathrm{tr}\, R(x;{\bf t};2\xi;\epsilon)R(x;{\bf t};2\nu;\epsilon)-1)}{(\xi-\nu)^2}=\sum_{i, j\ge 0} \frac{1}{\xi^{i+2}\nu^{j+2}}\frac{\Omega_{i,j}(x;\bm{s};\epsilon)}{2^{i+j}}.\label{=1}		
\end{align}
Then by using \eqref{todatauf1}, we have
\begin{align}
      \Omega_{i,j}^{\rm NLS}(x;{\bf t};\epsilon)=\Omega_{i,j}(x;{\bf t};\epsilon)=\epsilon^2\frac{\partial^2 \log\tau(x;{\bf t};\epsilon)}{\partial t_{i}\partial t_{j}},\quad i,j\ge 0,\label{=2}
\end{align}
 where $\tau(x;{\bf t};\epsilon)$ is the tau function of the solution 
 $(v(x;{\bf t};\epsilon),w(x;{\bf t};\epsilon))$ to the Toda lattice hierarchy. Then, by using~\eqref{DZtauf}, for any $x$,  $\tau(x;{\bf t};\epsilon)$ is the tau function for the NLS hierarchy of the solution $(q(x;{\bf t};\epsilon),r(x;{\bf t};\epsilon))$. The theorem is proved.
\end{proof}

We are now to prove Theorem~\ref{part2}.

\begin{proof}[Proof of Theorem~\ref{part2}]
Let $q_{\rm s}(X;\epsilon):=q(X, {\bf 0};\epsilon),r_{\rm s}(X;\epsilon):=r(X, {\bf 0};\epsilon)$. Define
\begin{align}\label{defVWI}
		V_{\rm s}(X;\epsilon):=\epsilon\partial_{X}({\rm log\,}q_{\rm s}(X;\epsilon)),\quad W_{\rm s}(X;\epsilon):=q_{\rm s}(X;\epsilon)r_{\rm s}(X;\epsilon).
\end{align}
Let $(\tilde{v}^*(x;X;\epsilon),\tilde{w}^*(x;X;\epsilon))$ 
be the unique solution to equations~\eqref{todaeq1} with the initial data
\begin{align}\label{xtodaeqini}
\tilde{v}^*(x;X;\epsilon)|_{x=0}=V_{\rm s}(X;\epsilon),\quad\tilde{w}^*(x;X;\epsilon)|_{x=0}=W_{\rm s}(X;\epsilon).
\end{align}
Let $(v,w)=(\tilde{v}(x;{\bf t};\epsilon),\tilde{w}(x;{\bf t};\epsilon))$ be the unique power-series-in-$t_{>0}$ solution to the Toda lattice hierarchy~\eqref{Todahierarchy} with the initial data
\begin{align}\label{initialtilvw}
\tilde{v}(x;{\bf t};\epsilon)|_{t_{>0}=0}=\tilde{v}^*(x;X;\epsilon),\quad\tilde{w}(x;{\bf t};\epsilon)|_{t_{>0}=0}=\tilde{w}^*(x;X;\epsilon),
\end{align} 

We are to show that the following graph commutes:
\begin{center}
\begin{tikzcd}
		(q({\bf t};\epsilon),r({\bf t},\epsilon)) \arrow[r,"\eqref{defqr*}"] \arrow[d, "t_{>0}=0"] & (V({\bf t};\epsilon),W({\bf t};\epsilon)) \arrow[rr, "{\eqref{todaeq1},\eqref{todaeqini2}}"] & & (v(x;{\bf t};\epsilon),w(x;{\bf t};\epsilon)) \\
		(q_{\rm s}(X;\epsilon),r_{\rm s}(X;\epsilon)) \arrow[r, "\eqref{defVWI}"]  & (V_{\rm s}(X;\epsilon),W_{\rm s}(X;\epsilon)) \arrow[rr, "{\eqref{todaeq1}, \eqref{xtodaeqini}}"] & & (\tilde{v}^*(x;X;\epsilon),\tilde{w}^*(x;X;\epsilon)) \arrow[u, "\star"]        
\end{tikzcd}
\end{center}
Here, $\star$ means to solve the Toda lattice hierarchy ($t_{>0}$--flows).
	
Define $q=q(x;{\bf t};\epsilon)$ by \eqref{defw} and \eqref{defqx} (cf.~\cite{Y20}) 
and define $r(x;{\bf t};\epsilon):=1/q(x-\epsilon;{\bf t};\epsilon)$. 
Then we have
\begin{align}\label{tilqrvw}
    	q(x;{\bf t};\epsilon)r(x;{\bf t};\epsilon)=\tilde{w}(x;{\bf t};\epsilon),\quad\epsilon\partial_{X}({\rm log\,}q(x;{\bf t};\epsilon))=\tilde{v}(x;{\bf t};\epsilon).
\end{align} 
By using \eqref{tilqrvw},  \eqref{initialtilvw}, \eqref{xtodaeqini} and \eqref{defVWI}, we have
\begin{align}
	\partial_{X}({\rm log\,}q(x;{\bf t};\epsilon))|_{x=0,t_{>0}=0}=\partial_{X}({\rm log\,}q_{\rm s}(X;\epsilon)),\quad
	(q(x;{\bf t};\epsilon)r(x;{\bf t};\epsilon))|_{x=0,t_{>0}=0}=q_{\rm s}(X;\epsilon)r_{\rm s}(X;\epsilon).
\end{align}
Noting that $q(x;{\bf t};\epsilon)$ is unique up to a function of $\epsilon$, we can take a suitable function of $\epsilon$ such that
\begin{align}\label{qqIrrI}
	q(x;{\bf t};\epsilon)|_{x=0,t_{>0}=0}=q_{\rm s}(X;\epsilon),\quad r(x;{\bf t};\epsilon))|_{x=0,t_{>0}=0}=r_{\rm s}(X;\epsilon).
\end{align}
Then, by using Theorem \ref{part1},  $(q(0;{\bf t};\epsilon),r(0;{\bf t};\epsilon))$ is a solution to the NLS hierarchy with the initial value \eqref{qqIrrI}.
Recall that $(q({\bf t};\epsilon),r({\bf t};\epsilon))$ is the unique solution  to the NLS hierarchy with the initial value  \eqref{todaeq11}.
Then we have
\begin{align}\label{tilqq*}
       q(0;{\bf t};\epsilon)=q({\bf t};\epsilon),\quad r(0;{\bf t};\epsilon)=r({\bf t};\epsilon).
\end{align}
Then by using \eqref{tilqrvw}, \eqref{tilqq*} and \eqref{defqr*}, we have
\begin{align}\label{tilvw}
	\tilde{v}(x;{\bf t};\epsilon)|_{x=0}=V({\bf t};\epsilon),\quad
	\tilde{w}(x;{\bf t};\epsilon)|_{x=0}=W({\bf t};\epsilon).
\end{align}
Then $(\tilde{v}(x;{\bf t};\epsilon),\tilde{w}(x;{\bf t};\epsilon))$ is a solution to \eqref{todaeq1} 
with the initial value \eqref{tilvw}. By the uniqueness for the solution of equations~\eqref{todaeq1} with the initial value \eqref{tilvw},
	$(v(x;{\bf t};\epsilon),w(x;{\bf t};\epsilon))=(\tilde{v}(x;{\bf t};\epsilon),\tilde{w}(x;{\bf t};\epsilon))$ is a solution to the Toda lattice hierarchy.

We are now to prove the compatibility between \eqref{deftau1}, \eqref{deftau2} and \eqref{deftau3}.
By using~\eqref{matrixrelation} and \eqref{l32}, we have
\begin{align}
    	R^{\rm{NLS}}(x+\epsilon;{\bf t};\xi;\epsilon)\mathcal{U}^{\rm NLS}(x;{\bf t};\xi;\epsilon)-\mathcal{U}^{\rm NLS}(x;{\bf t};\xi;\epsilon)R^{\rm{NLS}}(x;{\bf t};\xi;\epsilon)=0,\label{rx+}
\end{align}
where $\mathcal{U}^{\rm NLS}(x;{\bf t};\xi;\epsilon):=\begin{pmatrix}
	\epsilon\partial_{X}\left(\log q(x;{\bf t};\epsilon)\right)-2\xi&-q(x;{\bf t};\epsilon)\\
	\frac{1}{q(x;{\bf t};\epsilon)}&0
\end{pmatrix}$.
Then, by using \eqref{rx+}, the proof of the compatibility between \eqref{deftau1}, \eqref{deftau2} and \eqref{deftau3} is similar to that in \cite{DY17}. We omit its details. 
Then by using \eqref{=1}, we have
\begin{align}
&\epsilon^2\frac{\partial^2 \log\tau^{\rm NLS}(x;{\bf t};\epsilon)}{\partial t_{i}\partial t_{j}}=\Omega_{i,j}(x;{\bf t};\epsilon),\quad i,j\ge  0.\label{identitytodatau1}
\end{align}
This together with \eqref{deftau2},~\eqref{deftau3} and~\eqref{tilqrvw} implies that
$\tau^{\rm NLS}(x;{\bf t};\epsilon)$ is the tau-function of $(v(x;{\bf t};\epsilon),w(x;{\bf t};\epsilon))$ to the Toda lattice hierarchy. The theorem is proved.
\end{proof}

In an upcoming publication, we generalize the Carlet--Dubrovin--Zhang theorem (Theorems \ref{part1}, \ref{part2}) 
by using the MR method
to the constraint KP hierarchy~\cite{C92, CZ94, LZZ} and prove the conjecture given in~\cite{LZZ}. 

\begin{remark}\label{powerseries}
The assumption that $\epsilon\log q(\bt;\e)$ and $q(\bt;\e) r(\bt;\e)$ are power series of~$\epsilon$ from Theorem~\ref{part2} 
corresponds to generic classes of solutions to the Toda lattice 
hierarchy. Indeed, let us construct 
the solutions to the Toda lattice hierarchy starting with certain initial data for the NLS hierarchy.
Let $q_{\rm s}(X;\epsilon),r_{{\rm s}}(X;\epsilon)$ be two given functions of $X,\epsilon$ such that 
both $\epsilon\log q_{\rm s}(X,\e)$ and the product $q_{\rm s}(X,\e) r_{\rm s}(X,\e)$ 
are power series of~$\epsilon$.
Let $(q({\bf t};\epsilon),r({\bf t};\epsilon))$ be the unique solution to the NLS hierarchy specified by 
the initial condition: 
\begin{align}
	q({\bf t};\epsilon)|_{t_{>0}=0}=q_{\rm s}(X;\epsilon), \quad r({\bf t};\epsilon)|_{t_{>0}=0}=r_{\rm s}(X;\epsilon),\label{todaeq11}
\end{align}
where we recall that $X=t_{0}$. It follows from certain simple degree arguments that 
$\epsilon\log q(\bt;\e)$ and $q(\bt;\e) r(\bt;\e)$ are power series of~$\epsilon$. 
Then the corresponding NLS tau-function $\tau^{\rm NLS}(\bt;\e)$ of $(q,r)$ has the property that 
$\e^2 \log \tau^{\rm NLS}(\bt;\e)$ is a power series of~$\e$.  
Define $V({\bf t};\epsilon),W({\bf t};\epsilon)$ by~\eqref{defqr*}, and we know that 
$V({\bf t}; \epsilon), W({\bf t};\epsilon)$ are power series of~$\epsilon$. 
Write 
\begin{align}
& v(x;{\bf t};\epsilon) = \sum_{k\geq0} v^{[k]}({x;\bf t})\e^k, \quad w(x;{\bf t}; \epsilon) = \sum_{k\geq0} w^{[k]}(x;{\bf t})\e^k.
\end{align}
Similarly,  write also $V_{\rm s}(X;\epsilon) = \sum_{k\geq0} V_{\rm s}^{[k]}(X)\e^k$, $W_{\rm s}(X; \epsilon) = \sum_{k\geq0} W_{\rm s}^{[k]}(X)\e^k$.
By Taylor expanding~\eqref{todaeq1} in~$\epsilon$, we find that~\eqref{todaeq1} is equivalent to 
a sequence of first order evolutionary PDEs for $v^{[k]}, w^{[k]}$, $k\geq0$, with $x$ being the time
variable: 
\begin{align}
&	\partial_{x}(w^{[0]})=\partial_{X}(v^{[0]}),\quad w^{[0]}\partial_{x}(v^{[0]})=\partial_{X}( w^{[0]}),\\
&	\partial_{x}(w^{[1]})=\partial_{X}(v^{[1]})-\frac{1}{2}\partial^{2}_{x}(w^{[0]}),\quad w^{[0]}\partial_{x}(v^{[1]})=\partial_{X}(w^{[1]})-w^{[1]}\partial_{x}(v^{[0]})+ \frac{1}{2}w^{[0]}\partial^{2}_{x}(v^{[0]}),\\
& \dots.\nn
\end{align}
If either $W^{[0]}_{\rm s}(0)\neq0$, or $W^{[0]}_{\rm s}(0)=0$ and $(V_{\rm s}^{[0]})'(0)\neq 0$, then 
by using the above first order evolutionary PDEs and 
by induction we know that equations~\eqref{todaeq1} 
with the initial data $V({\bf t}; \epsilon), W({\bf t};\epsilon)$
have a unique power-series-in-$(x,t_{>0})$ solution.  
This indicates that the solutions under consideration 
to the Toda lattice hierarchy are generic (cf.~also~\cite{D09,DZ,DZ04,Z02}) and could be applied in many occasions (cf.~\cite{A21,BR21,DY17,DYZ20,DZ04,Ma,Z02}).
The above argument also explains that the assumption for the solvability of \eqref{todaeq1}, \eqref{todaeqini2} 
in Theorem~\ref{part2}
is true for wide classes of solutions.
\end{remark}

\section{A simple algorithm of computing correlators in hermitian matrix models}\label{section4}

In this section, using Theorems \ref{part1}, \ref{part2} and the MR method, we give a simple algorithm of computing correlators in hermitian matrix models.
 
Let $n$ be a positive integer, and $\rho=\rho(x)$ a given function. Let  
$Z_{n}({\bf s};\epsilon)$ be the following partition function of hermitian matrix model of 
size $n$:
\begin{align}\label{partitionex2}
Z_{n}({\bf s};\epsilon)=\frac{(2\pi)^{-n}\epsilon^{-\frac{1}{12}}}{\mathrm{Vol}(n)}\int_{\mathcal{H}(n)} \rho(M) \, e^{-\frac{1}{\epsilon}\mathrm{tr}\,V(M)}dM.
\end{align}
Here, 
$\mathcal{H}(n)$ denotes the space of $n\times n$ hermitian matrices, $M=(M_{ij})_{n\times n}\in \mathcal{H}(n)$, ${\bf s}=(s_1,s_2,s_3,\cdots)$, 
$V(M)=\frac{1}{2}M^2-\sum_{j\ge 1} s_jM^{j}$, $dM=\prod_{i=1}^{n}dM_{ii}\prod_{i<j} d\mathrm{Re}(M_{ij})d\mathrm{Im}(M_{ij})$, and 
\begin{align}
\mathrm{Vol}(n):=\frac{\pi^{\frac{n(n-1)}{2}}}{G(n+1)}
\end{align}
with $G$ being the Barnes $G$-function defined by $G(n+1):=\prod_{k=1}^{n-1}k!$.
We also denote $Z_0({\bf s};\epsilon):=\e^{-1/12}$. 

Let $x=n\epsilon$ (the 't~Hooft coupling constant~\cite{H74-1,H74-2}), and define 
\begin{align}\label{defzvw}
	v(x;{\bf s};\epsilon):=\epsilon\frac{\partial}{\partial X}\log\frac{Z_{n+1} ({\bf s};\epsilon)}{Z_{n} ({\bf s};\epsilon)}\bigg|_{n=\frac{x}{\epsilon}},\quad w(x;{\bf s};\epsilon):=\frac{Z_{n+1} ({\bf s};\epsilon)Z_{n-1} ({\bf s};\epsilon)}{Z_{n} ({\bf s};\epsilon)^2}\bigg|_{n=\frac{x}{\epsilon}}.
\end{align}
Using the theory of orthogonal polynomials~\cite{D99} (cf.~also~\cite{DY17, M91}), we know that
$(v(x;{\bf s};\epsilon),w(x;{\bf s};\epsilon))$ is a particular solution to the Toda lattice hierarchy \eqref{Todahierarchy} 
and $Z_{x/\e} ({\bf s};\epsilon)$ gives the tau-function of this solution~\cite{DY17} with $t_{j}=s_{j+1},j=0,1,2,\dots$.
Define
	\begin{align}\label{qrZ}
		q(x;{\bf t};\epsilon)=\frac{Z_{n+1}({\bf s};\epsilon)}{Z_{n}({\bf s};\epsilon)}\bigg|_{n=\frac{x}{\epsilon},s_{j+1}=t_{j},j\ge 0},\quad r(x;{\bf t};\epsilon)=\frac{Z_{n-1}({\bf s};\epsilon)}{Z_{n}({\bf s};\epsilon)}\bigg|_{n=\frac{x}{\epsilon},s_{j+1}=t_{j},j\ge 0}.
	\end{align}
From Theorem \ref{part1}, we know that for any~$x$, the pair $(q(x;{\bf t};\epsilon),r(x;{\bf t};\epsilon))$ is a particular solution to the NLS hierarchy \eqref{Tnlshierarchy} 
and $Z_{x/\e} ({\bf s};\epsilon)$ gives the tau-function of this solution. 
An algorithm of computing the logarithmic derivatives 
of $\log Z_n({\bf s};\epsilon)$ based on~\eqref{kpointtoda} is given in~\cite{DY17}. Let us give 
a brief review. The first step of this algorithm is to compute the initial value of $(v,w)$. 
According to Proposition~A.2.3 of~\cite{DY17} (cf.~also~\cite{AvM95, UT}) we know that  
$Z_{n}=Z_{n} ({\bf s};\epsilon)$ satisfies the Toda equation:
\begin{align}
\left(\Lambda+\Lambda^{-1}-2\right) \left(\log Z_{n}\right)=\log\left( \e^2 \frac{\partial^2 \log Z_{n}}{\partial X^2}\right).\label{secondeq}
\end{align}
Suppose that one can compute the integral $Z_{1}(X,0,\dots;\epsilon)$.	
Then from $Z_{1}(X,0,\dots;\epsilon)$ and $Z_{0}(X,0,\dots;\epsilon)$, 
in principle one can compute $Z_{n}(X, 0, \dots;\epsilon)$ by using~\eqref{secondeq}. 
The initial value $(v(x,{\bf 0};\epsilon),w(x,{\bf 0};\epsilon))$ could then be obtained 
from~\eqref{defzvw}. The second step is to compute the initial basic matrix-resolvent~\eqref{todamr}--\eqref{l33}
for the Toda lattice hierarchy. The correlators can then be computed from~\eqref{kpointtoda}.

Let us now apply Theorems \ref{part1}, \ref{part2} to provide an improvement of the first step of 
the above algorithm. Note that for the hermitian matrix models, we have 
$q_{\rm s}(X;\epsilon)=\e^{1/12}Z_{1}(X,0,\dots;\epsilon)$, 
$r_{\rm s}(X;\epsilon)=0$, and 
\begin{align}
&w_{\rm s}(X;\epsilon)=q_{\rm s}(X;\epsilon) r_{\rm s}(X;\epsilon)\equiv0, \quad 
v_{\rm s}(X;\epsilon):=\epsilon\partial_{X}(\log q_{\rm s}(X;\epsilon)).\label{appeq4}	
\end{align}
One can then solve equations~\eqref{todaeq1} with the initial condition 
\begin{align}
w(x;X;\epsilon)|_{x=0}=w_{\rm s}(X;\epsilon),\quad v(x;X;\epsilon)|_{x=0}=v_{\rm s}(X;\epsilon), \label{inivalue}
\end{align}
and obtain $w(x;X;\e),v(x;X;\e)$ (cf.~Remark~\ref{powerseries}). 
Then the initial value for the solution under consideration to the Toda lattice hierarchy is obtained 
by $v(x;{\bf 0};\epsilon)=v(x;X;\epsilon)|_{X=0},w(x;{\bf 0};\epsilon)=w(x;X;\epsilon)|_{X=0}$. 
We note that the feature of the improvement for the algorithm is the following: one is to solve the system of first-order difference equations~\eqref{todaeq1} instead of 
to solve the second-order difference equation~\eqref{secondeq}, and the system~\eqref{todaeq1} 
can be recast into a sequence of first-order evolutionary PDEs, recursively, as explained in Remark~\ref{powerseries} and Theorem~\ref{part2}. The improved algorithm works for the GUE~\cite{DY17}, LUE~\cite{GGR20} and JUE~\cite{GGR21} 
solutions.  Let us explain in more details the simple algorithm by means of examples, 
which include GUE (see Example~1) below.

\paragraph{Example 1}
Let $Z_{n}({\bf s};\epsilon)$ denote the GUE partition function of size~$n$ (cf.~\cite{AvM95,BIZ80,DiFGZ,D99,DY17,FIK92,HZ86,KKN99,M91,Z18-1,Z18-2}), i.e.,
\begin{align}
	Z_{n}({\bf s};\epsilon)=\frac{(2\pi)^{-n}\epsilon^{-\frac{1}{12}}}{\mathrm{Vol}(n)}\int_{\mathcal{H}(n)} e^{-\frac{1}{\epsilon}\mathrm{tr}\,V(M)}dM.\label{partitionfunction}
\end{align}
The logarithmic derivatives of $Z_n$ at ${\bf s}={\bf 0}$, 
denoted by $\langle\mathrm{tr} M^{i_1}\cdots\mathrm{tr} M^{i_k}\rangle_c(n)$, 
are called {\it connected GUE correlators}. They have important relations with enumeration of ribbon graphs~\cite{BIZ80,DY17,HZ86,H74-1,H74-2}. 

We have in this case $Z_{1} (X,0,\dots;\epsilon)=(2\pi)^{-\frac{1}{2}}\epsilon^{\frac{5}{12}}e^{\frac{ X^2}{2\epsilon}}$. Then 
$q_{\rm s}(X;\epsilon)=(2\pi)^{-\frac{1}{2}}\epsilon^{\frac{1}{2}}e^{\frac{X^2}{2\epsilon}}$,
$r_{\rm s}(X;\epsilon)=0$,
$v_{\rm s}(X;\epsilon)=X$, $w_{\rm s}(X;\epsilon)=0$.
Solving~\eqref{todaeq1}, we obtain $v(x;X;\epsilon)\equiv X$ and $w(x;X;\epsilon) 
\equiv x$. 
So the initial value for the GUE solution to the Toda lattice hierarchy is given by 
 $v(x;{\bf 0};\epsilon)=0,w(x;{\bf 0};\epsilon)=x$, 
which agrees with the literature~\cite{D99, DY17, M91}. Also in this case $q(x;X;\epsilon)\equiv\epsilon^{\frac{1}{2}+n}\frac{n!}{\sqrt{2\pi}}e^{\frac{X^2}{2\epsilon}}$,  
 $r(x;X;\epsilon)\equiv\epsilon^{\frac{1}{2}-n}\frac{\sqrt{2\pi}}{(n-1)!}e^{-\frac{X^2}{2\epsilon}}$.

In~\cite{DY17}, B.~Dubrovin and the second author of the present paper computed the explicit expression for the 
corresponding initial matrix resolvent:
\begin{align}
R(x;\lambda;\epsilon):=\begin{pmatrix}
1&0\\
0&0
\end{pmatrix}+\epsilon n\sum_{k=0}^{\infty}\epsilon^k\frac{(2k-1)!!}{\lambda^{2k+2}}\begin{pmatrix}
		(2k+1)e_{n,2k+1}&-\lambda g_{n+1,2k}\\
		\frac{\lambda}{\epsilon n}g_{n,2k}&-(2k+1)e_{n,2k+1}
\end{pmatrix}\label{Rseries},
\end{align}
where
\begin{align}
	e_{n,2k+1}&=\frac{1}{n}\sum_{i=0}^{k}2^i\binom{k}{i}\binom{n}{i+1}={}_2F_1(-k,1-n; 2; 2), \label{hypgeo1}\\
	g_{n,2k}&=\sum_{i=0}^{k}2^i\binom{k}{i}\binom{n-1}{i}={}_2F_1(-k,1-n; 1; 2).\label{hypgeo3}
\end{align}
Here, 
\begin{align}
	{}_2F_1(a,b;c;z)=\sum_{j=0}^{\infty}\frac{(a)_j(b)_j}{(c)_j}\frac{z^j}{j!}=1+\frac{ab}{c}\frac{z}{1!}+\frac{a(a+1)b(b+1)}{c(c+1)}\frac{z^2}{2!}+\cdots 
\end{align}
denotes the Gauss hypergeometric function with $(a)_i$ being the increasing Pochhammer symbol, i.e., $(a)_i=a(a+1)\cdots(a+i-1)$. 
The following theorem can then be obtained as a result of this computation and formula~\eqref{kpointtoda}.
\begin{theorem}(\cite{DY17})\label{appkeythm}
For any $k\ge 2$, the generating series for connected $k$-point GUE correlators has the expression
	\begin{align}
		C_k(n;\lambda_1,...,\lambda_k,\epsilon)=-\sum_{\sigma\in S_k/C_k} \frac{\mathrm{tr}\,\left[ R(x;\lambda_{\sigma(1)},\epsilon)\cdots R(x;\lambda_{\sigma(k)},\epsilon)\right]}{(\lambda_{\sigma(1)}-\lambda_{\sigma(2)})\cdots(\lambda_{\sigma(k-1)}-\lambda_{\sigma(k)})(\lambda_{\sigma(k)}-\lambda_{\sigma(1)})}-\frac{\delta_{k,2}}{(\lambda_1-\lambda_2)^2}.\label{appthm2}
	\end{align}
Here, $C_k(n;\lambda_1,...,\lambda_k,\epsilon)$ denotes the following generating series
\begin{align}
C_k(n;\lambda_1,...,\lambda_k,\epsilon):=\epsilon^{k}\sum_{i_1,...,i_k=1}^{\infty}\frac{\langle\mathrm{tr} M^{i_1}\cdots\mathrm{tr} M^{i_k}\rangle_c(n)}{\lambda_{1}^{i_1+1}\cdots\lambda_{k}^{i_k+1}} \quad (k\ge 1).
\end{align}
\end{theorem}

We note that for $k=1$, one can further apply the string equation~\cite{DY17} and obtain 
\begin{align}
		C_1(n;\lambda;\epsilon)=n\sum_{k=1}^{\infty}\epsilon^k\frac{(2k-1)!!}{\lambda^{2k+1}}e_{n,2k+1}.\label{onepoint}
\end{align} 
This formula is equivalent to the Harer--Zagier formula given in the earlier work~\cite{HZ86}. For $k=2$, 
formula~\eqref{appthm2} is equivalent to the formula derived by Morozov and Shakirov~\cite{MS09}.

Following~\cite{DYZ20}, let us now consider the analytic (as opposed to formal) expression for the matrix resolvent $R(x;\lambda;\epsilon)$.
Denote by $\He_{n}(z)$ the Hermite function \cite{WG10} (cf.~also~\cite{WW63}), defined as the contour integral
\begin{align}
	\He_{n}(z):=-\frac{\Gamma(n+1)}{2\pi {\rm i}}\int_{C}e^{-zt-\frac{t^2}{2}}(-t)^{-n-1} dt,\quad |\arg (-t)|< \pi,
\end{align}
where $\Gamma$ denotes the Gamma function, and 
$C$ is a contour on the complex $t$-plane that starts at `infinity' on the real axis, encircles the origin in the positive direction and returns to the starting point (see Figure~\ref{figure}).
\begin{figure}
	\centering
	\begin{tikzpicture}[thick]
		\draw[-Straight Barb] (-1.5,0)--(3.2,0);
		\draw[-Straight Barb] (0,-1.2)--(0,1.2);
		\draw[-Straight Barb] [domain=15:165] plot ({0.6*cos(\x)}, {0.6*sin(\x)});
		\draw[-] [domain=165:345] plot ({0.6*cos(\x)}, {0.6*sin(\x)});
		\draw[-Straight Barb]  (3,0.15529)--(2,0.15529);
		\draw[-]  (2,0.15529)--(0.57955,0.15529);
		\draw[-Straight Barb]  (0.57955,-0.15529)--(1.5,-0.15529);
		\draw[-]  (1.5,-0.15529)--(3,-0.15529);
		\coordinate (0) at (0,0);  
		\node[below left] at (0){{\tiny 0}}; 
	\end{tikzpicture}
\caption{The integration contour $C$ on the complex $t$-plane} \label{figure}
\end{figure}
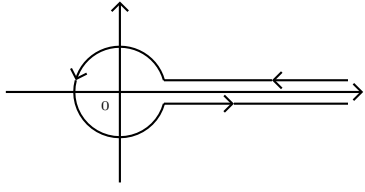
For $n\in\CC$, the Hermite function $\He_{n}(z)$ is an analytic function on the complex $z$-plane, which has the following asymptotic expansion within the sector $|{\rm arg}\,z|<\frac{3\pi}{4}$ as $|z|\to \infty$:
\begin{align}
		\He_n(z)\sim z^n\sum_{k=0}^{\infty}\frac{(-1)^k \, (n-2k+1)_{2k}}{2^k \, k! \, z^{2k}}.\label{paraexpand}
\end{align}
It follows that for fixed $n,\epsilon\in\mathbb{C}$, and for  
$-\frac{3\pi}{4}<{\rm arg}\Big(\frac{\lambda}{\sqrt{\epsilon}}\Big)<\frac{\pi}{4}$, we have, as $\lambda\to \infty$,
\begin{align}
		&\epsilon^{\frac{n}{2}}\He_{n}\big(\frac{\lambda}{\sqrt{\epsilon}}\big)\sim\sum_{j\ge 0} (-1)^{j}\frac{(n-2j+1)_{2j}}{2^jj!\lambda^{2j}}\epsilon^{j}\lambda^{n}=:\psi_{A}(n,\lambda),\\
		&\Gamma(n+1)\epsilon^{\frac{n+1}{2}}{\rm i}^{n+1}\lambda\He_{-n-1}\big({\rm i}\frac{\lambda}{\sqrt{\epsilon}}\big)\sim \Gamma(n+1)\sum_{j\ge 0} \frac{(n+1)_{2j}}{2^jj!\lambda^{2j}}\epsilon^{n+j+1}\lambda^{-n}=:\psi_{B}(n,\lambda).
\end{align}
Observe that the above defined $\psi_{A}(n,\lambda),\psi_{B}(n,\lambda)$ coincides with the particular pair of formal 
wave functions associated to the initial data $(0,\epsilon n )$ for the GUE solution to the Toda lattice hierarchy constructed in~\cite{Y20}. For any fixed $n$, define 
\begin{align}
&M(n;\lambda;\epsilon)\nonumber \\
:=&\frac{1}{\epsilon^{n}\lambda}\begin{pmatrix}
\epsilon^{\frac{n}{2}}\He_{n}\big(\frac{\lambda}{\sqrt{\epsilon}}\big)\Lambda^{-1}\left(\epsilon^{\frac{n+1}{2}}{\rm i}^{n+1}\lambda\He_{-n-1}\big({\rm i}\frac{\lambda}{\sqrt{\epsilon}}\big)\right)&-\epsilon^{\frac{n}{2}}\He_{n}\big(\frac{\lambda}{\sqrt{\epsilon}}\big) n\epsilon^{\frac{n+1}{2}}{\rm i}^{n+1}\lambda\He_{-n-1}\big({\rm i}\frac{\lambda}{\sqrt{\epsilon}}\big)\\
\Lambda^{-1}\left(\epsilon^{\frac{n}{2}}\He_{n}\big(\frac{\lambda}{\sqrt{\epsilon}}\big)\right)\Lambda^{-1}\left(\epsilon^{\frac{n+1}{2}}{\rm i}^{n+1}\lambda\He_{-n-1}\big({\rm i}\frac{\lambda}{\sqrt{\epsilon}}\big)\right)&-\Lambda^{-1}\left(\epsilon^{\frac{n}{2}}\He_{n}\big(\frac{\lambda}{\sqrt{\epsilon}}\big)\right) n\epsilon^{\frac{n+1}{2}}{\rm i}^{n+1}\lambda\He_{-n-1}\big({\rm i}\frac{\lambda}{\sqrt{\epsilon}}\big)
\end{pmatrix} \nonumber	\\
=&{\rm i}^n\begin{pmatrix}
\He_{n}\big(\frac{\lambda}{\sqrt{\epsilon}}\big)\He_{-n}\big({\rm i}\frac{\lambda}{\sqrt{\epsilon}}\big)&-{\rm i}\epsilon^{\frac{1}{2}}n\He_{n}\big(\frac{\lambda}{\sqrt{\epsilon}}\big)\He_{-n-1}\big({\rm i}\frac{\lambda}{\sqrt{\epsilon}}\big)\\
\epsilon^{-\frac{1}{2}}\He_{n-1}\big(\frac{\lambda}{\sqrt{\epsilon}}\big)\He_{-n}\big({\rm i}\frac{\lambda}{\sqrt{\epsilon}}\big)&-{\rm i}n \He_{n-1}\big(\frac{\lambda}{\sqrt{\epsilon}}\big)\He_{-n-1}\big({\rm i}\frac{\lambda}{\sqrt{\epsilon}}\big)
\end{pmatrix}.\label{defineM}
\end{align}
By using the Proposition 3 of~\cite{Y20}, we then arrive at the following proposition.

\begin{pro}\label{ayexp}
For fixed $n,\epsilon\in\mathbb{C}$, the asymptotic expansion of $M(n;\lambda;\epsilon)$ in all orders as $\lambda\to \infty$ within the sector $-\frac{3\pi}{4}<{\rm arg}\Big(\frac{\lambda}{\sqrt{\epsilon}}\Big)<\frac{\pi}{4}$ coincides with the formal power series $R(x;\lambda;\epsilon)$.
\end{pro} 

We now define analytic $k$-point functions $I_{k}(n;\lambda_{1},\dots,\lambda_{k};\epsilon)$ ($k\ge 2$; the case $k=1$ will be treated later) by
\begin{align}
	&I_{k}(n;\lambda_{1},\dots,\lambda_{k};\epsilon):=-\sum_{\sigma\in S_k/ C_k} \frac{\mathrm{tr}\,[M(n;\lambda_{\sigma(1)},\epsilon)\cdots M(n;\lambda_{\sigma(k)},\epsilon)]}{(\lambda_{\sigma(1)}-\lambda_{\sigma(2)})\cdots(\lambda_{\sigma(k-1)}-\lambda_{\sigma(k)})(\lambda_{\sigma(k)}-\lambda_{\sigma(1)})}-\frac{\delta_{k,2}}{(\lambda_{1}-\lambda_{2})^{2}}\label{jiexintaustr}.
\end{align} 
Using a similar argument as in the proof of the Proposition 2 of~\cite{DYZ20} that $I_{k}(n;\lambda_{1},\dots,\lambda_{k};\epsilon)$, $k\ge 2$, are analytic along the diagonals $\lambda_i=\lambda_{j}, i\neq j$. 
Formula~\eqref{appthm2} and Proposition~\ref{ayexp} then imply that 
 for any fixed $\epsilon\in\mathbb{C}$ 
the asymptotic expansion of $I_{k}(n;\lambda_{1},\dots,\lambda_{k};\epsilon)$ coincides with $C_k(n;\lambda_1,\dots,\lambda_{k};\epsilon)$ as $\lambda_i\to \infty$ within $-\frac{3\pi}{4}<{\rm arg}\Big(\frac{\lambda_{i}}{\sqrt{\epsilon}}\Big)<\frac{\pi}{4}$, $i=1,\dots,k$.
Define a meromorphic function $D(n;\lambda,\mu;\epsilon)$ by
\begin{align}
		&D(n;\lambda,\mu;\epsilon) \nonumber\\
		=&\frac{\epsilon^{\frac{n}{2}}\He_{n}\big(\frac{\lambda}{\sqrt{\epsilon}}\big)\Lambda^{-1}\left(\Gamma(n+1)\epsilon^{\frac{n+1}{2}}{\rm i}^{n+1}\mu\He_{-n-1}\big({\rm i}\frac{\mu}{\sqrt{\epsilon}}\big)\right)-\Lambda^{-1}\left(\epsilon^{\frac{n}{2}}\He_{n}\big(\frac{\lambda}{\sqrt{\epsilon}}\big)\right)\Gamma(n+1)\epsilon^{\frac{n+1}{2}}{\rm i}^{n+1}\mu\He_{-n-1}\big({\rm i}\frac{\mu}{\sqrt{\epsilon}}\big)}{\lambda-\mu}\label{defD}\\
		=&-\frac{{\rm i}^{n+1}\epsilon^{n}\Gamma(n)\mu}{\lambda-\mu}\Big(n\He_{n-1}\big(\frac{\mu}{\sqrt{\epsilon}}\big)\He_{-n-1}\big({\rm i}\frac{\lambda}{\sqrt{\epsilon}}\big)+{\rm i}\He_{n}\big(\frac{\mu}{\sqrt{\epsilon}}\big)\He_{-n}\big({\rm i}\frac{\lambda}{\sqrt{\epsilon}}\big)\Big).
\end{align}

Using \eqref{defineM}, \eqref{jiexintaustr} and \eqref{defD},  we arrive at the following theorem. 
\begin{theorem}
The analytic functions $I_k,k\ge 2$ have the expressions
    \begin{align}
			I_{k}(n;\lambda_{1},\dots,\lambda_{k};\epsilon)
			=\frac{(-1)^{k-1}}{\epsilon^{kn}(\Gamma(n))^{k}\prod_{j=1}^{k}\lambda_{j}}\sum_{\sigma\in S_k/ C_k} \prod_{i=1}^{k} D(n;\lambda_{\sigma(i)},\lambda_{\sigma(i+1)};\epsilon)-\frac{\delta_{k,2}}{(\lambda_1-\lambda_2)^2}.
	\end{align}
\end{theorem}

Let us also define
\begin{align}
    I_1(n;\lambda;\epsilon):=\frac{{\rm i}^{n+1}n}{\epsilon}\int_{\lambda}^{+\infty} \He_{n-1}\big(\frac{\lambda}{\sqrt{\epsilon}}\big)\He_{-n-1}\big({\rm i}\frac{\lambda}{\sqrt{\epsilon}}\big)d\lambda.
\end{align}
Using~\eqref{paraexpand} and Proposition~\ref{ayexp}, we arrive at the following proposition.
\begin{pro}
	The function $I_{1}(n;\lambda;\epsilon)$ is asymptotic to  $C_{1}(n;\lambda;\epsilon)$ as $\lambda\to \infty$ in the sector $-\frac{3\pi}{4}<{\rm arg}\Big(\frac{\lambda}{\sqrt{\epsilon}}\Big)<\frac{\pi}{4}$ for any fixed $\epsilon\in\mathbb{C}$.
\end{pro}

\paragraph{Example 2}
Consider the following formal power series $Z_{n} ({\bf s};\sigma,\epsilon)$ defined by
\begin{align}
	Z_{n}({\bf s};\sigma,\epsilon)=\frac{(2\pi)^{-n}\epsilon^{-\frac{1}{12}}}{\mathrm{Vol}(n)}\int_{\mathcal{H}(n)} \det(1-\sigma M)e^{-\frac{1}{\epsilon}\mathrm{tr}\,V(M)}dM.
\end{align}
When $n=1$, we have 
\begin{align}
& Z_{1} (X,0,\dots;\sigma,\epsilon)=\frac{\epsilon ^{5/12} (1-\sigma  X) e^{\frac{X^2}{2 \epsilon }}}{\sqrt{2 \pi }}\label{jue1z1},\\
& q_{\rm s}(X;\sigma,\epsilon)=\frac{\epsilon ^{\frac{1}{2}} (1-\sigma  X) e^{\frac{X^2}{2 \epsilon }}}{\sqrt{2 \pi }},\quad r_{\rm s}(X;\sigma,\epsilon)=0,\quad w_{\rm s}(X;\sigma,\epsilon)=0,\quad v_{\rm s}(X;\sigma,\epsilon)=X-\frac{\epsilon\sigma}{1-X\sigma}.
\end{align}
Then by solving~\eqref{todaeq1}
we get
\begin{align}
&w(x;X;\sigma,\epsilon)=\epsilon n\frac{\h_{n-1}(X_{1})\h_{n+1}(X_{1})}{(\h_{n}(X_{1}))^2},\\
&v(x;X;\sigma,\epsilon)=X+\frac{\sqrt{2\epsilon}n\h_{n-1}(X_{1})}{\h_{n+1}(X_{1})}-\frac{\sqrt{2\epsilon}(n+1)\h_{n}(X_{1})}{\h_{n+1}(X_{1})},
\end{align}
where $X_{1}=\frac{1-\sigma X}{\sqrt{2\epsilon}\sigma}$ and $\h_{n}(z):=2^{\frac{n}{2}}\He_{n}(\sqrt{2}z)$. Using \eqref{defw}--\eqref{defqx}, we have 
\begin{align}
	&q(x;X,0,\dots;\sigma,\epsilon)=\frac{\sigma\epsilon^{n+1}}{2\sqrt{\pi}}n!e^{\frac{X^2}{2\epsilon}}\frac{\h_{n+1}(X_{1})}{\h_{n}(X_{1})},\\
	&r(x;X,0,\dots;\sigma,\epsilon)=\frac{2\sqrt{\pi}}{\sigma\epsilon^{n}(n-1)!}e^{-\frac{X^2}{2\epsilon}}\frac{\h_{n-1}(X_{1})}{\h_{n}(X_{1})}.
\end{align}
Setting $X=0$, we obtain
\begin{align*}
w(x;0;\sigma,\epsilon)=&\epsilon n\frac{\h_{n-1}(X_{2})\h_{n+1}(X_{2})}{\h_{n}^2(X_{2})}\\
=&x-\epsilon\left(\sigma^2x+3\sigma^4x^2+10\sigma^6x^3+\mathcal{O}\left(x ^4\right)\right)+\epsilon^2\left(3\sigma^4x-25\sigma^6x^2+154\sigma^8x^3+\mathcal{O}\left(x ^4\right)\right)+\mathcal{O}\left(\epsilon ^3\right),\\
v(x;0;\sigma,\epsilon)=&\frac{\sqrt{2\epsilon}n\h_{n-1}(X_{2})}{\h_{n+1}(X_{2})}-\frac{\sqrt{2\epsilon}(n+1)\h_{n}(X_{2})}{\h_{n+1}(X_{2})}\\
=&-\epsilon\left(\sigma+2\sigma^3x+6\sigma^5x^2+\mathcal{O}\left(x ^3\right)\right)+\epsilon^2\left(4\sigma^5x+36\sigma^7x^2+232\sigma^9x^3+\mathcal{O}\left(x ^4\right)\right)+\mathcal{O}\left(\epsilon ^3\right)
\end{align*}
where $X_{2}=\frac{1}{\sqrt{2\epsilon}\sigma}$. Using \eqref{todatauf1}--\eqref{todatauf3} and \eqref{jue1z1}, we find
\begin{align}
Z_{n}(X,0,\dots;\sigma,\epsilon)=2^{-n} \pi ^{-\frac{n}{2}} \sigma^n \epsilon^{\frac{n^2+n}{2}-\frac{1}{12}} e^{\frac{n X^2}{2 \epsilon }}G(n+1)\h_{n}(X_{1}).
\end{align}
This formula agrees with Br\'ezin--Hikami's computation~\cite{BH00}. 
Using~\eqref{todarec1}--\eqref{intitoda} one can compute the initial 
matrix resolvent~\eqref{todamr}, denoted by~$R(x;\lambda;\sigma,\epsilon)$. We have 
\begin{align*}
R(x;\lambda;\sigma,\epsilon)=&\begin{pmatrix}
		1&0\\
		0&0
	\end{pmatrix}
+\begin{pmatrix}
		0&-\epsilon n\frac{\h_{n-1}(X_{2})\h_{n+1}(X_{2})}{\h_{n}^2(X_{2})}\\
		1&0
\end{pmatrix}
\frac{1}{\lambda}\\
&{\tiny+\begin{pmatrix}
\epsilon n\frac{\h_{n-1}(X_{2})\h_{n+1}(X_{2})}{(\h_{n}(X_{2}))^{2}}&-\sqrt{2\epsilon}\epsilon n\frac{\h_{n-1}(X_{2})(n\h_{n-1}(X_{2})\h_{n+1}(X_{2})-(n+1)(\h_{n}(X_{2}))^2)}{(\h_{n}(X_{2}))^{3}}\\
\sqrt{2\epsilon}\frac{(n-1)\h_{n}(X_{2})\h_{n-2}(X_{2})-(\h_{n-1}(X_{2}))^2}{\h_{n}(X_{2})\h_{n-1}(X_{2})}&-\epsilon n\frac{\h_{n-1}(X_{2})\h_{n+1}(X_{2})}{(\h_{n}(X_{2}))^2}
\end{pmatrix}\frac{1}{\lambda^2}+\cdots.}
\end{align*}
The logarithmic derivatives of $Z_{n} ({\bf s};\sigma,\epsilon)$ at ${\bf s}={\bf 0}$, 
denoted by $\langle\mathrm{tr} M^{i_1}\cdots\mathrm{tr} M^{i_k}\rangle_c(n)$, 
can then be computed by using \eqref{todatauf1}, \eqref{kpointtoda}. We list the first few of them here:
\begin{align*}
& \langle(\mathrm{tr} M)^2\rangle_c(n)=\frac{n}{\epsilon}\frac{\h_{n-1}(X_{2})\h_{n+1}(X_{2})}{(\h_{n}(X_{2}))^{2}},\\
& \langle\mathrm{tr} M \mathrm{tr} M^2\rangle_c(n)=\frac{\sqrt{2}n}{\sqrt{\epsilon}}\frac{(n-1)\h_{n-2}(X_{2})\h_{n+1}(X_{2})-(n+1)\h_{n-1}(X_{2})\h_{n}(X_{2})}{(\h_{n}(X_{2}))^2},\\
&\langle(\mathrm{tr} M)^3\rangle_c(n)=\frac{\sqrt{2}n}{\sqrt{\epsilon}}\bigg(-\frac{(n+1)\h_{n+1}(X_{2})}{\h_{n}(X_{2})}-\frac{(n-1)\h_{n-2}(X_{2})\h_{n+1}(X_{2})}{(\h_{n}(X_{2}))^2}+\frac{2n(\h_{n-1}(X_{2}))^2\h_{n+1}(X_{2})}{(\h_{n}(X_{2}))^3}\bigg).
\end{align*}

\paragraph{Example 3}
Consider the following formal power series $Z_{n}({\bf s};\rho,\epsilon)$ defined by
\begin{align}
	Z_{n}({\bf s};\rho,\epsilon)=\frac{(2\pi)^{-n}\epsilon^{-\frac{1}{12}}}{\mathrm{Vol}(n)}\int_{\mathcal{H}(n)} \det(1-\rho M^2)e^{-\frac{1}{\epsilon}\mathrm{tr}\,V(M)}dM.\label{partitionfunction2}
\end{align}
When $n=1$ we have 
\begin{align}
&Z_{1} (X,0,\dots;\rho,\epsilon)=\frac{e^{\frac{X^2}{2\epsilon}}\epsilon^{\frac{5}{12}}(1-(X^2+\epsilon)\rho)}{\sqrt{2\pi}}\label{jue2z1},\\
&q_{\rm s}(X;\rho,\epsilon)=\frac{e^{\frac{X^2}{2\epsilon}}\epsilon^{\frac{1}{2}}(1-(X^2+\epsilon)\rho)}{\sqrt{2\pi}},\quad r_{\rm s}(X;\rho,\epsilon)=0,\\
& w_{\rm s}(X;\rho,\epsilon)=0,\quad v_{\rm s}(X;\rho,\epsilon)=X+\frac{2\epsilon X\rho}{(X^2+\epsilon)\rho-1}.
\end{align}
Solving~\eqref{todaeq1} with the initial value \eqref{inivalue}, we get
\begin{align*}
&w(x;X;\rho,\epsilon)={\epsilon n\frac{(\h_{n-1} (X_{3} )\h_{n} (X_{4} )-\h_{n-1} (X_{4} )\h_{n} (X_{3} ))(\h_{n+1} (X_{3} )\h_{n+2} (X_{4} )-\h_{n+1} (X_{4} )\h_{n+2} (X_{3} ) )}{ (\h_{n} (X_{3} )\h_{n+1} (X_{4} )-\h_{n} (X_{4} )\h_{n+1} (X_{3} ) )^2}},\\
&v(x;X;\rho,\epsilon)
	={(n+1)\frac{ (X\h_{n+1} (X_{3} )-\sqrt{2\epsilon}\h_{n} (X_{3} ) )\h_{n+2} (X_{4} )+ (X\h_{n+1} (X_{4} )-\sqrt{2\epsilon}\h_{n} (X_{4} ) )\h_{n+2} (X_{3} )}{\h_{n+1} (X_{3} )\h_{n+2} (X_{4} )-\h_{n+1} (X_{4} )\h_{n+2} (X_{3} )}}\\
&\qquad\qquad\qquad\,\, {-n\frac{ (X\h_{n} (X_{3} )-\sqrt{2\epsilon}\h_{n-1} (X_{3} ) )\h_{n+1} (X_{4} )+ (X\h_{n} (X_{4} )-\sqrt{2\epsilon}\h_{n-1} (X_{4} ) )\h_{n+1} (X_{3} )}{\h_{n} (X_{3} )\h_{n+1} (X_{4} )-\h_{n} (X_{4} )\h_{n+1} (X_{3} )}},
\end{align*}
where $X_{3}=\frac{1-\sqrt{\rho}X}{\sqrt{2\epsilon\rho}},X_{4}=\frac{-1-\sqrt{\rho}X}{\sqrt{2\epsilon\rho}}$. Using \eqref{defw}--\eqref{defqx}, we have 
\begin{align}
		&q(x;X,0,\dots;\rho,\epsilon)=\frac{\rho\epsilon^{\frac{3}{2}+n}}{2\sqrt{2\pi}}e^{\frac{X^2}{2\epsilon}}n!\frac{\h_{n+1}(X_{3})\h_{n+2}(X_{4})-\h_{n+1}(X_{4})\h_{n+2}(X_{3})}{\h_{n}(X_{3})\h_{n+1}(X_{4})-\h_{n}(X_{4})\h_{n+1}(X_{3})},\\
		&r(x;X0,\dots;\rho,\epsilon)=\frac{2\sqrt{2\pi}}{\rho\epsilon^{\frac{1}{2}+n}(n-1)!}e^{-\frac{X^2}{2\epsilon}}\frac{\h_{n-1}(X_{3})\h_{n}(X_{4})-\h_{n-1}(X_{4})\h_{n}(X_{3})}{\h_{n}(X_{3})\h_{n+1}(X_{4})-\h_{n}(X_{4})\h_{n+1}(X_{3})}.
\end{align}
Setting $ X=0$, we obtain
\begin{align*}
w(x;0;\rho,\epsilon)=&\epsilon n\frac{\h_{n-1}(X_{5})\h_{n+2}(X_{5})}{\h_{n}(X_{5})\h_{n+1}(X_{5})}\\
		=& x -\epsilon\bigl(2\rho x+6\rho^2x^2+20\rho^3x^3+\mathcal{O}\left(x ^4\right)\bigr)+\epsilon^2\bigl(4\rho^2x+36\rho^3x^2+232\rho^4x^3+\mathcal{O}\left(x^4\right)\bigr)+\mathcal{O}\left(\epsilon ^3\right),\\
		v(x;0;\rho,\epsilon)=&0,
\end{align*}
where $X_{5}=\frac{1}{\sqrt{2\epsilon\rho}}$. Using \eqref{todatauf1}--\eqref{todatauf3} and \eqref{jue2z1}, we find
\begin{align*}
	Z_{n}(X,0,\dots;\rho,\epsilon)
	=(-1)^{n+1} 2^{-\frac{3}{2}(n+1)} \pi ^{-\frac{n}{2}} \rho ^{n+\frac{1}{2}} \epsilon ^{\frac{n^2}{2}+n+\frac{5}{12}} e^{\frac{n X^2}{2 \epsilon }}G(n+1)\Big(\h_{n}(X_{3})\h_{n+1}(X_{4})-\h_{n}(X_{4})\h_{n+1}(X_{3})\Big).
\end{align*}
This formula again agrees with Br\'ezin--Hikami's computation~\cite{BH00}. 
Using~\eqref{todarec1}--\eqref{intitoda} one can compute the initial matrix resolvent~\eqref{todamr}, denoted by~$R(x;\lambda;\rho,\epsilon)$.
We have
\begin{align*}
	&R(x;\lambda;\rho,\epsilon)
	\\
	=&\begin{pmatrix}
		1&0\\
		0&0
	\end{pmatrix}+\begin{pmatrix}
		0&-\epsilon n\frac{\h_{n-1}(X_{5})\h_{n+2}(X_{5})}{\h_{n}(X_{5})\h_{n+1}(X_{5})}\\
		1&0
	\end{pmatrix}\frac{1}{\lambda}
	+\begin{pmatrix}
		\epsilon n\frac{\h_{n-1}(X_{5})\h_{n+2}(X_{5})}{\h_{n}(X_{5})\h_{n+1}(X_{5})}&0\\
		0&-\epsilon n\frac{\h_{n-1}(X_{5})\h_{n+2}(X_{5})}{\h_{n}(X_{5})\h_{n+1}(X_{5})}
	\end{pmatrix}\frac{1}{\lambda^2}+\cdots.
\end{align*}
The logarithmic derivatives of $Z_{n} ({\bf s};\rho,\epsilon)$ at ${\bf s}={\bf 0}$, 
denoted by $\langle\mathrm{tr} M^{i_1}\cdots\mathrm{tr} M^{i_k}\rangle_c(n)$, 
can then be computed by using \eqref{todatauf1}, \eqref{kpointtoda}. We list the first few of them here:
\begin{align*}
&\langle \left(\tr M\right)^2\rangle_c(n)=\frac{n}{\epsilon}\frac{\h_{n-1}(X_{5})\h_{n+2}(X_{5})}{\h_{n}(X_{5})\h_{n+1}(X_{5})},\\
&\langle\tr M \tr M^{3}\rangle_c(n)
=\frac{(n-1)\h_{n-2}(X_{5})\h_{n+2}(X_{5})}{(\h_{n}(X_{5}))^2}+n\Big(\frac{\h_{n-1}(X_{5})\h_{n+2}(X_{5})}{\h_{n}(X_{5})\h_{n+1}(X_{5})}\Big)^2+\frac{(n+1)\h_{n-1}(X_{5})\h_{n+3}(X_{5})}{(\h_{n+1}(X_{5}))^2},\\
&\langle\tr M^{2}\left(\tr M\right)^{2}  \rangle_c(n)=\frac{n(n+1)\h_{n-1}(X_{5})\h_{n+3}(X_{5})}{\epsilon(\h_{n+1}(X_{5}))^2}-\frac{n(n-1)\h_{n-2}(X_{5})\h_{n+2}(X_{5})}{\epsilon(\h_{n}(X_{5}))^2}.
\end{align*}

\smallskip

\paragraph{Acknowledgments}
We would like to thank Professors Mattia Cafasso, Xian Liao, Dangzheng Liu for helpful suggestions. 
The work is supported by the National Key R 
and D Program of China 2020YFA0713100 
and NSFC 12061131014.

\begin{appendix}
\section{Review of the MR method to tau-functions for the Toda lattice hierarchy}\label{appa}

Let us give a brief review of the MR method of computing logarithmic derivatives of tau-functions for the Toda lattice hierarchy. 
Denote by $\mathbb{Z}\left[{\bf v},{\bf w}\right]$ of polynomial ring, where ${\bf v}=(v=v_{0},v_{1},v_{-1},v_{2},v_{-2},\dots),{\bf w}=(w=w_{0},w_{1},w_{-1},w_{2},w_{-2},\dots)$. 
Define the shift operator $\Lambda:\mathbb{Z}\left[{\bf v},{\bf w}\right]\to \mathbb{Z}\left[{\bf v},{\bf w}\right]$ via
\begin{align}
	\Lambda( v_{k})=v_{k+1},\quad \Lambda (w_{k})=w_{k+1},\quad \Lambda (fg)=\Lambda(f)\Lambda(g)
\end{align}
for any $ k\in\mathbb{Z}$ and $f,g\in \mathbb{Z}\left[{\bf v},{\bf w}\right]$. Denote by $\mathcal{L}(\lambda)$ the matrix Lax operator for the Toda lattice:
\begin{align}\label{laxtoda}
	\mathcal{L}(\lambda)=
\Lambda+U(\lambda),\quad U(\lambda):=\begin{pmatrix}
		v-\lambda&w\\
		-1&0
	\end{pmatrix}.
\end{align}
Here, $\lambda$ is a parameter. 
It is proved in~\cite{DY17} that there exists a unique series $R(\lambda)$ satisfying
\begin{align}
&R(\lambda)-\begin{pmatrix}
1&0\\0&0\\
\end{pmatrix}
\in\mathrm{Mat} \left(2, \mathbb{Z}\left[{\bf v},{\bf w}\right][\epsilon]\left[\left[\lambda^{-1} \right]\right]\lambda^{-1}\right),\label{todamr}\\
&\Lambda(R(\lambda))U(\lambda)-U(\lambda)R(\lambda)=0, \label{l32}\\
&\mathrm{tr} \, R(\lambda)=1,\quad \mathrm{det}\ R(\lambda)=0. \label{l33}
\end{align}
The unique $R(\lambda)$ is called the {\it basic matrix resolvent} of~$\mathcal{L}(\lambda)$.
Write
\begin{align}
	&R(\lambda)=\begin{pmatrix}
		1+\alpha(\lambda)&\beta(\lambda)\\
		\gamma(\lambda)&\alpha(\lambda)
	\end{pmatrix},  \label{Rleadingterm1023} \\
	&\alpha(\lambda)=\sum_{j\ge 0} \frac{a_{ j}}{\lambda^{j+1}},\quad 
	\beta(\lambda)=\sum_{j\ge 0} \frac{b_{j}}{\lambda^{j+1}},\quad  
	\gamma(\lambda)=\sum_{j\ge 0} \frac{c_{j}}{\lambda^{j+1}},\label{compoR}
\end{align}
where $a_{j},b_{j},c_{j}\in\mathbb{Z}\left[{\bf v},{\bf w}\right]$. 
 Then $a_{j},b_{j},c_{j},j\ge0$, satisfy that 
\begin{align}
	& 		b_{j}=-wc_{j},\quad c_{j+1}=v_{-1}c_{j}+a_{j}+\Lambda^{-1}(a_{j}),\label{todarec1}\\
	& 		a_{j+1}-\Lambda(a_{j+1})+v(\Lambda(a_{j})-a_{j})+w_{1}\Lambda^{2}(c_{j})-wc_{j}=0,\\
	& 		a_{l}=\sum_{i+j=l-1} (wc_{i}\Lambda(c_{ j})-a_{i}a_{j}),\quad l\ge 1, \label{sumtoda}\\
	&       a_{0}=0,\quad  c_{0}=1\label{intitoda}.
\end{align}
The first few terms of $R(\lambda)$ are given by
\begin{align}
	R(\lambda)=\begin{pmatrix}
		1&0\\
		0&0
	\end{pmatrix}+\begin{pmatrix}
		0&-w\\
		1&0
	\end{pmatrix}\frac{1}{\lambda}+\begin{pmatrix}
		w&-vw\\
		v_{-1}&-w
	\end{pmatrix}\frac{1}{\lambda^2}+\cdots.
\end{align}
Define a sequence of derivations on $\epsilon^{-1}\mathbb{Z}\left[{\bf v},{\bf w}\right]$ by 
\begin{align}\label{absfloetoda}
		\mathcal{D}_j(v)=\epsilon^{-1}(\Lambda(a_{j+1})-a_{j+1}),\quad
		\mathcal{D}_j(w)=\epsilon^{-1}w(\Lambda(c_{j+1})-c_{j+1}),
	    \quad j\ge 0,
\end{align}
as well as requiring $\left[\mathcal{D}_j,\Lambda\right]=0$.
We know that from~\cite{DY17,Y20} $(\mathcal{D}_j)_{j\ge 0}$ all commute.
If we think of~$v,w$, as two functions $v(x),w(x)$ of $x$, respectively, and $v_{i},w_{i}$ as $v(x+i\epsilon),w(x+i\epsilon)$, 
then the Toda lattice hierarchy can be written as
\begin{align}\label{Todahierarchy}
		\frac{\partial v(x)}{\partial t_{j}}=\mathcal{D}_j(v)(x),\quad
		 \frac{\partial w(x)}{\partial t_{j}}=\mathcal{D}_j(w)(x),
\end{align}
where $j\ge 0$, and the $\mathcal{D}_j(v)(x),\mathcal{D}_j(w)(x)$ are defined as $\mathcal{D}_j(v)(x),\mathcal{D}_j(w)(x)$ with $v_{i},w_{i}$ replaced by $v(x+i\epsilon),w(x+i\epsilon)$, respectively. For $j\ge 0$, define
\begin{align}
	V _{j}(\lambda):=\Big(\lambda^{j+1}R(\lambda)\Big)_{+}+\begin{pmatrix}
		0&0\\
		0&c_{j+1}
	\end{pmatrix}. \label{defVj1023}
\end{align}
It is proved in~\cite{DY17} that
\begin{align}
	\epsilon\mathcal{D}_{j}(R(\lambda))=[V _{j}(\lambda),R(\lambda)],\quad j\ge 0.
\end{align}
In particular, when $j=0$, we have
\begin{align}
\epsilon\mathcal{D}_0(R(\lambda))=[V _0(\lambda),R(\lambda)],\quad V_{0}(\lambda)=\begin{pmatrix}
	\lambda&-w\\
	1&v_{-1}
\end{pmatrix}. \label{35}
\end{align}
Using~\eqref{35} one can obtain
\begin{align}
\epsilon \mathcal{D}_{0}(c_{j})=2a_{j}+v_{-1}c_{j}-c_{j+1}.\label{dcj}
\end{align}
Define a collection of polynomials  $(\Omega_{i,j})_{ i,j\ge 0}$ in $\mathbb{Z}\left[{\bf v},{\bf w}\right]$ by
\begin{gather}
	\frac{\mathrm{tr}\, R(\lambda)R(\mu)-1}{(\lambda-\mu)^2}=\sum_{i, j\ge 0} \frac{\Omega_{i;j}}{\lambda^{i+2}\mu^{j+2}}.  \label{taustr}
\end{gather}
For example, $\Omega_{0,0}=w$,  $\Omega_{1,0}=\Omega_{0,1}=w(v+v_{1})$.
These $\Omega_{i,j}$ satisfy
\begin{align}
	\Omega_{i;j}\in \mathbb{Z}\left[{\bf v}, {\bf w}\right],\quad \Omega_{i,j}=\Omega_{j,i},\quad \mathcal{D}_l( \Omega_{i,j})=\mathcal{D}_i(\Omega_{l,j}),\qquad \forall\, i,j,l\ge 0. \label{taustr1018}
\end{align}
The polynomials $(\Omega_{i;j})_{i,j\ge 0}$ are called the {\it tau-structrue} of the Toda lattice hierarchy.
For an arbitrary solution $v=v(x; {\bf t};\epsilon), w=w(x;{\bf t};\epsilon)$ to the Toda lattice hierarchy~\eqref{Todahierarchy}, 
let $R(x;{\bf t};\lambda;\epsilon)$ be the associated matrix resolvent. We know from~\eqref{taustr1018} that 
there exists a function $\tau(x;{\bf t};\epsilon)$ such that
\begin{align}
&	\sum_{i, j\ge 0} \frac{1}{\lambda^{i+2}\mu^{j+2}}\epsilon^2\frac{\partial^2\log\tau(x;{\bf t};\epsilon)}{\partial t_{i}\partial t_{j}}=
	\frac{\mathrm{tr}R(x;{\bf t};\lambda;\epsilon)R(x;{\bf t};\mu;\epsilon)-1}{(\lambda-\mu)^2}\label{todatauf1},\\
&	\frac{1}{\lambda}+\sum_{i\ge 0} \frac{1}{\lambda^{i+2}}\epsilon\frac{\partial}{\partial t_{i}}\log\frac{\tau(x+\epsilon; {\bf t};\epsilon)}{\tau(x;{\bf t};\epsilon)}=\left[\Lambda(R(x;{\bf t};\lambda;\epsilon))\right]_{21} \label{211},\\
&	\frac{\tau(x+\epsilon; {\bf t};\epsilon)\tau(x-\epsilon; {\bf t};\epsilon)}{\tau(x;{\bf t};\epsilon)^2}=w(x;{\bf t};\epsilon).\label{todatauf3}
\end{align}
The function $\tau(x;{\bf t};\epsilon)$ is determined uniquely by the solution 
$(v(x;{\bf t};\epsilon),w(x;{\bf t};\epsilon))$ up to 
multiplying by the exponential of a linear function 
\begin{align}
\tau(x;{\bf t};\epsilon)\mapsto e^{b_0+b_{1}x+\sum_{j\ge 0} b_{j+2}t_{j}}\tau(x;{\bf t};\epsilon), \quad b_{0},b_{1},b_{2},\cdots\in\mathbb{C}((\epsilon)).
\end{align}
The function $\tau(x;{\bf t};\epsilon)$ defined by \eqref{todatauf1}--\eqref{todatauf3} is called the {\it tau-function of the solution $(v(x;{\bf t};\epsilon),w(x;{\bf t};\epsilon))$} to the Toda lattice hierarchy. Remarkably, the higher logarithmic derivatives of $\tau(x;{\bf t};\epsilon)$ can all be 
expressed in terms of the matrix resolvent as follows~\cite{DY17}:
\begin{align}\label{kpointtoda}
\sum_{i_1,\dots,i_k\ge 0} \epsilon^k\frac{\partial^k\log\tau(x;{\bf t};\epsilon)}{\partial t_{i_1}\cdots\partial t_{i_k}}\prod_{j=1}^{k}\frac{1}{\lambda_{j}^{i_j+2}}
	=-\sum_{\sigma\in S_k/ C_k}\frac{\mathrm{tr} \left(R(x;{\bf t};\lambda_{\sigma(1)};\epsilon)\cdots R(x;{\bf t};\lambda_{\sigma(k)};\epsilon)\right)}{(\lambda_{\sigma(1)}-\lambda_{\sigma(2)})\cdots(\lambda_{\sigma(k-1)}-\lambda_{\sigma(k)})(\lambda_{\sigma(k)}-\lambda_{\sigma(1)})}.
\end{align}

\end{appendix}

\medskip

\noindent School of Mathematical Sciences, USTC, Hefei 230026, P.R. China

\smallskip

\noindent fuang@mail.ustc.edu.cn, diyang@ustc.edu.cn

\end{document}